\documentclass[acmsmall]{acmart}

\newif\iffull
\fulltrue

\usepackage{amsmath}

\usepackage{amssymb}
\usepackage{array, multirow}
\usepackage[plain,Fig]{algorithm}
\usepackage{algorithmicx}
\usepackage[noend]{algpseudocode}
\usepackage{comment}
\usepackage{enumitem}
\usepackage{graphicx}
\usepackage[final]{listings}
\usepackage{pifont}
\usepackage{relsize}
\usepackage{stmaryrd}
\usepackage{syntax}
\usepackage{textcomp}
\usepackage{xspace}
\usepackage{wrapfig}

\usepackage{float}

\usepackage{amsthm}

\usepackage{booktabs}
\usepackage{multirow}

\usepackage{hhline}
\usepackage{upgreek}

\usepackage{fancyvrb}

\usepackage{subcaption}

\usepackage{varwidth}
\usepackage{makecell}
\usepackage{placeins}

\makeatletter
\def\PY@reset{\let\PY@it=\relax \let\PY@bf=\relax%
    \let\PY@ul=\relax \let\PY@tc=\relax%
    \let\PY@bc=\relax \let\PY@ff=\relax}
\def\PY@tok#1{\csname PY@tok@#1\endcsname}
\def\PY@toks#1+{\ifx\relax#1\empty\else%
    \PY@tok{#1}\expandafter\PY@toks\fi}
\def\PY@do#1{\PY@bc{\PY@tc{\PY@ul{%
    \PY@it{\PY@bf{\PY@ff{#1}}}}}}}
\def\PY#1#2{\PY@reset\PY@toks#1+\relax+\PY@do{#2}}

\expandafter\def\csname PY@tok@gd\endcsname{\def\PY@tc##1{\textcolor[rgb]{0.63,0.00,0.00}{##1}}}
\expandafter\def\csname PY@tok@gu\endcsname{\let\PY@bf=\textbf\def\PY@tc##1{\textcolor[rgb]{0.50,0.00,0.50}{##1}}}
\expandafter\def\csname PY@tok@gt\endcsname{\def\PY@tc##1{\textcolor[rgb]{0.00,0.27,0.87}{##1}}}
\expandafter\def\csname PY@tok@gs\endcsname{\let\PY@bf=\textbf}
\expandafter\def\csname PY@tok@gr\endcsname{\def\PY@tc##1{\textcolor[rgb]{1.00,0.00,0.00}{##1}}}
\expandafter\def\csname PY@tok@cm\endcsname{\let\PY@it=\textit\def\PY@tc##1{\textcolor[rgb]{0.25,0.50,0.50}{##1}}}
\expandafter\def\csname PY@tok@vg\endcsname{\def\PY@tc##1{\textcolor[rgb]{0.10,0.09,0.49}{##1}}}
\expandafter\def\csname PY@tok@vi\endcsname{\def\PY@tc##1{\textcolor[rgb]{0.10,0.09,0.49}{##1}}}
\expandafter\def\csname PY@tok@vm\endcsname{\def\PY@tc##1{\textcolor[rgb]{0.10,0.09,0.49}{##1}}}
\expandafter\def\csname PY@tok@mh\endcsname{\def\PY@tc##1{\textcolor[rgb]{0.40,0.40,0.40}{##1}}}
\expandafter\def\csname PY@tok@cs\endcsname{\let\PY@it=\textit\def\PY@tc##1{\textcolor[rgb]{0.25,0.50,0.50}{##1}}}
\expandafter\def\csname PY@tok@ge\endcsname{\let\PY@it=\textit}
\expandafter\def\csname PY@tok@vc\endcsname{\def\PY@tc##1{\textcolor[rgb]{0.10,0.09,0.49}{##1}}}
\expandafter\def\csname PY@tok@il\endcsname{\def\PY@tc##1{\textcolor[rgb]{0.40,0.40,0.40}{##1}}}
\expandafter\def\csname PY@tok@go\endcsname{\def\PY@tc##1{\textcolor[rgb]{0.53,0.53,0.53}{##1}}}
\expandafter\def\csname PY@tok@cp\endcsname{\def\PY@tc##1{\textcolor[rgb]{0.74,0.48,0.00}{##1}}}
\expandafter\def\csname PY@tok@gi\endcsname{\def\PY@tc##1{\textcolor[rgb]{0.00,0.63,0.00}{##1}}}
\expandafter\def\csname PY@tok@gh\endcsname{\let\PY@bf=\textbf\def\PY@tc##1{\textcolor[rgb]{0.00,0.00,0.50}{##1}}}
\expandafter\def\csname PY@tok@ni\endcsname{\let\PY@bf=\textbf\def\PY@tc##1{\textcolor[rgb]{0.60,0.60,0.60}{##1}}}
\expandafter\def\csname PY@tok@nl\endcsname{\def\PY@tc##1{\textcolor[rgb]{0.63,0.63,0.00}{##1}}}
\expandafter\def\csname PY@tok@nn\endcsname{\let\PY@bf=\textbf\def\PY@tc##1{\textcolor[rgb]{0.00,0.00,1.00}{##1}}}
\expandafter\def\csname PY@tok@no\endcsname{\def\PY@tc##1{\textcolor[rgb]{0.53,0.00,0.00}{##1}}}
\expandafter\def\csname PY@tok@na\endcsname{\def\PY@tc##1{\textcolor[rgb]{0.49,0.56,0.16}{##1}}}
\expandafter\def\csname PY@tok@nb\endcsname{\def\PY@tc##1{\textcolor[rgb]{0.00,0.50,0.00}{##1}}}
\expandafter\def\csname PY@tok@nc\endcsname{\let\PY@bf=\textbf\def\PY@tc##1{\textcolor[rgb]{0.00,0.00,1.00}{##1}}}
\expandafter\def\csname PY@tok@nd\endcsname{\def\PY@tc##1{\textcolor[rgb]{0.67,0.13,1.00}{##1}}}
\expandafter\def\csname PY@tok@ne\endcsname{\let\PY@bf=\textbf\def\PY@tc##1{\textcolor[rgb]{0.82,0.25,0.23}{##1}}}
\expandafter\def\csname PY@tok@nf\endcsname{\def\PY@tc##1{\textcolor[rgb]{0.00,0.00,1.00}{##1}}}
\expandafter\def\csname PY@tok@si\endcsname{\let\PY@bf=\textbf\def\PY@tc##1{\textcolor[rgb]{0.73,0.40,0.53}{##1}}}
\expandafter\def\csname PY@tok@s2\endcsname{\def\PY@tc##1{\textcolor[rgb]{0.73,0.13,0.13}{##1}}}
\expandafter\def\csname PY@tok@nt\endcsname{\let\PY@bf=\textbf\def\PY@tc##1{\textcolor[rgb]{0.00,0.50,0.00}{##1}}}
\expandafter\def\csname PY@tok@nv\endcsname{\def\PY@tc##1{\textcolor[rgb]{0.10,0.09,0.49}{##1}}}
\expandafter\def\csname PY@tok@s1\endcsname{\def\PY@tc##1{\textcolor[rgb]{0.73,0.13,0.13}{##1}}}
\expandafter\def\csname PY@tok@dl\endcsname{\def\PY@tc##1{\textcolor[rgb]{0.73,0.13,0.13}{##1}}}
\expandafter\def\csname PY@tok@ch\endcsname{\let\PY@it=\textit\def\PY@tc##1{\textcolor[rgb]{0.25,0.50,0.50}{##1}}}
\expandafter\def\csname PY@tok@m\endcsname{\def\PY@tc##1{\textcolor[rgb]{0.40,0.40,0.40}{##1}}}
\expandafter\def\csname PY@tok@gp\endcsname{\let\PY@bf=\textbf\def\PY@tc##1{\textcolor[rgb]{0.00,0.00,0.50}{##1}}}
\expandafter\def\csname PY@tok@sh\endcsname{\def\PY@tc##1{\textcolor[rgb]{0.73,0.13,0.13}{##1}}}
\expandafter\def\csname PY@tok@ow\endcsname{\let\PY@bf=\textbf\def\PY@tc##1{\textcolor[rgb]{0.67,0.13,1.00}{##1}}}
\expandafter\def\csname PY@tok@sx\endcsname{\def\PY@tc##1{\textcolor[rgb]{0.00,0.50,0.00}{##1}}}
\expandafter\def\csname PY@tok@bp\endcsname{\def\PY@tc##1{\textcolor[rgb]{0.00,0.50,0.00}{##1}}}
\expandafter\def\csname PY@tok@c1\endcsname{\let\PY@it=\textit\def\PY@tc##1{\textcolor[rgb]{0.25,0.50,0.50}{##1}}}
\expandafter\def\csname PY@tok@fm\endcsname{\def\PY@tc##1{\textcolor[rgb]{0.00,0.00,1.00}{##1}}}
\expandafter\def\csname PY@tok@o\endcsname{\def\PY@tc##1{\textcolor[rgb]{0.40,0.40,0.40}{##1}}}
\expandafter\def\csname PY@tok@kc\endcsname{\let\PY@bf=\textbf\def\PY@tc##1{\textcolor[rgb]{0.00,0.50,0.00}{##1}}}
\expandafter\def\csname PY@tok@c\endcsname{\let\PY@it=\textit\def\PY@tc##1{\textcolor[rgb]{0.25,0.50,0.50}{##1}}}
\expandafter\def\csname PY@tok@mf\endcsname{\def\PY@tc##1{\textcolor[rgb]{0.40,0.40,0.40}{##1}}}
\expandafter\def\csname PY@tok@err\endcsname{\def\PY@bc##1{\setlength{\fboxsep}{0pt}\fcolorbox[rgb]{1.00,0.00,0.00}{1,1,1}{\strut ##1}}}
\expandafter\def\csname PY@tok@mb\endcsname{\def\PY@tc##1{\textcolor[rgb]{0.40,0.40,0.40}{##1}}}
\expandafter\def\csname PY@tok@ss\endcsname{\def\PY@tc##1{\textcolor[rgb]{0.10,0.09,0.49}{##1}}}
\expandafter\def\csname PY@tok@sr\endcsname{\def\PY@tc##1{\textcolor[rgb]{0.73,0.40,0.53}{##1}}}
\expandafter\def\csname PY@tok@mo\endcsname{\def\PY@tc##1{\textcolor[rgb]{0.40,0.40,0.40}{##1}}}
\expandafter\def\csname PY@tok@kd\endcsname{\let\PY@bf=\textbf\def\PY@tc##1{\textcolor[rgb]{0.00,0.50,0.00}{##1}}}
\expandafter\def\csname PY@tok@mi\endcsname{\def\PY@tc##1{\textcolor[rgb]{0.40,0.40,0.40}{##1}}}
\expandafter\def\csname PY@tok@kn\endcsname{\let\PY@bf=\textbf\def\PY@tc##1{\textcolor[rgb]{0.00,0.50,0.00}{##1}}}
\expandafter\def\csname PY@tok@cpf\endcsname{\let\PY@it=\textit\def\PY@tc##1{\textcolor[rgb]{0.25,0.50,0.50}{##1}}}
\expandafter\def\csname PY@tok@kr\endcsname{\let\PY@bf=\textbf\def\PY@tc##1{\textcolor[rgb]{0.00,0.50,0.00}{##1}}}
\expandafter\def\csname PY@tok@s\endcsname{\def\PY@tc##1{\textcolor[rgb]{0.73,0.13,0.13}{##1}}}
\expandafter\def\csname PY@tok@kp\endcsname{\def\PY@tc##1{\textcolor[rgb]{0.00,0.50,0.00}{##1}}}
\expandafter\def\csname PY@tok@w\endcsname{\def\PY@tc##1{\textcolor[rgb]{0.73,0.73,0.73}{##1}}}
\expandafter\def\csname PY@tok@kt\endcsname{\def\PY@tc##1{\textcolor[rgb]{0.69,0.00,0.25}{##1}}}
\expandafter\def\csname PY@tok@sc\endcsname{\def\PY@tc##1{\textcolor[rgb]{0.73,0.13,0.13}{##1}}}
\expandafter\def\csname PY@tok@sb\endcsname{\def\PY@tc##1{\textcolor[rgb]{0.73,0.13,0.13}{##1}}}
\expandafter\def\csname PY@tok@sa\endcsname{\def\PY@tc##1{\textcolor[rgb]{0.73,0.13,0.13}{##1}}}
\expandafter\def\csname PY@tok@k\endcsname{\let\PY@bf=\textbf\def\PY@tc##1{\textcolor[rgb]{0.00,0.50,0.00}{##1}}}
\expandafter\def\csname PY@tok@se\endcsname{\let\PY@bf=\textbf\def\PY@tc##1{\textcolor[rgb]{0.73,0.40,0.13}{##1}}}
\expandafter\def\csname PY@tok@sd\endcsname{\let\PY@it=\textit\def\PY@tc##1{\textcolor[rgb]{0.73,0.13,0.13}{##1}}}

\def\PYZob{\char`\{}
\def\PYZcb{\char`\}}

\def\PYZdl{\char`\$}

\makeatother

\definecolor{indigo}{HTML}{283593}
\definecolor{red}{HTML}{D50000}
\definecolor{green}{HTML}{558B2F}
\definecolor{orange}{HTML}{FF9800}
\definecolor{lightgreen}{HTML}{C5E1A5}
\definecolor{lightlime}{HTML}{B2FF59}
\definecolor{lightindigo}{HTML}{9FA8DA}
\definecolor{lightpurple}{HTML}{E1BEE7}
\definecolor{lightblue}{HTML}{90CAF9}
\definecolor{lightorange}{HTML}{FFB74D}
\definecolor{light-gray}{gray}{0.6}
\definecolor{super-light-gray}{gray}{0.77}

\makeatletter
\newcommand{\oset}[3][0ex]{%
  \mathrel{\mathop{#3}\limits^{
      \vbox to#1{\kern-2\ex@
        \hbox{$\scriptstyle#2$}\vss}}}}
\makeatother

\makeatletter
\newenvironment{btHighlight}[1][]
{\begingroup\tikzset{bt@Highlight@par/.style={#1}}\begin{lrbox}{\@tempboxa}}
{\end{lrbox}\bt@HL@box[bt@Highlight@par]{\@tempboxa}\endgroup}

\newcommand\btHL[1][]{%
  \begin{btHighlight}[#1]\bgroup\aftergroup\bt@HL@endenv%
}
\def\bt@HL@endenv{%
  \end{btHighlight}%
  \egroup
}
\newcommand{\bt@HL@box}[2][]{%
  \tikz[#1]{%
    \pgfpathrectangle{\pgfpoint{0.3pt}{0pt}}{\pgfpoint{\wd #2}{\ht #2}}%
    \pgfusepath{use as bounding box}%
    \node[anchor=base west,fill=lightorange,outer sep=0pt,inner xsep=0.3pt,inner ysep=0pt,minimum height=\ht\strutbox+0.3pt,#1]{\raisebox{0.3pt}{\strut}\strut\usebox{#2}};
  }%
}
\makeatother

\makeatletter
\let\OldStatex\Statex
\renewcommand{\Statex}[1][3]{%
  \setlength\@tempdima{\algorithmicindent}%
  \OldStatex\hskip\dimexpr#1\@tempdima\relax}
\makeatother

\newcommand{\irule}[2]%
{\mkern-2mu\displaystyle\frac{#1}{\vphantom{,}#2}\mkern-2mu}

\usepackage{tikz}
\usepackage{pgfplots}
\usetikzlibrary{calc,trees,positioning,arrows,chains,shapes.geometric,%
  decorations.pathreplacing,decorations.pathmorphing,decorations.text,shapes,%
  matrix,shapes.symbols,patterns,shadows}

\algnewcommand\algorithmicswitch{\textbf{switch}}
\algnewcommand\algorithmiccase{\textbf{case}}
\algnewcommand\algorithmicdefault{\textbf{default}}

\algdef{SE}[SWITCH]{Switch}{EndSwitch}[1]{\algorithmicswitch\ #1:}{\algorithmicend\ \algorithmicswitch}
\algdef{SE}[CASE]{Case}{EndCase}[1]{\algorithmiccase\ #1:}{\algorithmicend\ \algorithmiccase}%
\algdef{SE}[SWITCH]{Default}{EndDefault}{\algorithmicdefault:}{\algorithmicend\ \algorithmicdefault}%
\algtext*{EndSwitch}%
\algtext*{EndCase}%
\algtext*{EndDefault}

{\begin{list}{$\bullet$}{\setlength{\leftmargin}{1.5ex}%
\setlength{\itemindent}{.5ex}}
\setlength{\parindent}{2ex}
}%
{\end{list}}

\setcopyright{none}

\newcommand{\cmark}{\ding{51}}%
\newcommand{\xmark}{\ding{55}}%

\begin{document}

\title{Verifying Correct Usage of Context-Free API Protocols (Extended Version)
 }

\author{Kostas Ferles}

\affiliation{
  \institution{The University of Texas at Austin}            
  \country{USA}                    
}
\email{kferles@cs.utexas.edu}          

\author{Jon Stephens}

\affiliation{
  \institution{The University of Texas at Austin}           
  \country{USA}                   
}
\email{jon@cs.utexas.edu}         

\author{Isil Dillig}

\affiliation{
  \institution{The University of Texas at Austin}           
  \country{USA}                   
}
\email{isil@cs.utexas.edu}         

\ccsdesc[500]{Software and its engineering~General programming languages}
\ccsdesc[300]{Social and professional topics~History of programming languages}


\newcommand{\program}{\mathcal{P}}
\newcommand{\cfgrammar}{\mathcal{G}}
\newcommand{\cfgprog}{\cfgrammar_\program}
\newcommand{\cfgspec}{\cfgrammar_S}
\newcommand{\lang}{\mathcal{L}}
\newcommand{\toolname}{{\sc CFPChecker}} 
\newcommand{\todo}[1]{{\color{red} \bf #1}}
\newcommand{\mcode}[1]{{\footnotesize \mathsf{#1}}\xspace}
\newcommand{\code}[1]{{\small \texttt{#1}}}

\newcommand{\vars}{\mathit{V}}
\newcommand{\fields}{\mathit{F}}
\newcommand{\pstate}{\mathcal{\sigma}}
\newcommand{\ptrace}{\tau}

\newcommand{\preds}{\mathit{Pred}}
\newcommand{\cfa}{\mathcal{A}}
\newcommand{\cfas}{\mathit{PCFA}}
\newcommand{\meths}{\mathit{M}}
\newcommand{\meth}{\mathit{m}}
\newcommand{\stmts}{\mathit{\Sigma}}
\newcommand{\states}{\mathit{S}}
\newcommand{\initst}{\mathit{I}}
\newcommand{\transrel}{\delta}
\newcommand{\finalst}{\mathit{F}}
\newcommand{\s}{\mathit{s}}
\newcommand{\plocs}{\mathit{Loc}}
\newcommand{\ploc}{\mathit{l}}
\newcommand{\pred}{\varphi}
\newcommand{\main}{\mathit{main}}
\newcommand{\ccubes}{\textsf{CompleteCubes}}

\newcommand{\satloc}[2]{#1 {\shortdownarrow} {#2}}
\newcommand{\locof}{\mathit{Loc}}
\newcommand{\predof}{\mathit{Pred}}
\newcommand{\inof}{\mathit{In}}
\newcommand{\outof}{\mathit{Out}}
\newcommand{\initsof}{\mathit{Entry}}
\newcommand{\finsof}{\mathit{Exit}}
\newcommand{\transof}{\mathsf{Trans}}

\newcommand{\terminals}{\mathit{T}}
\newcommand{\nonterms}{\mathit{N}}
\newcommand{\rules}{\mathit{R}}
\newcommand{\startsym}{\mathit{S}}

\newcommand{\locrule}{\mathcal{L}}
\newcommand{\methrule}{\mathcal{M}}

\newcommand{\wildcardsof}{\mathcal{W}}
\newcommand{\typeor}{\Gamma}
\newcommand{\terminalsof}{\mathcal{T}}
\newcommand{\match}{\mathit{InstTerm}}
\newcommand{\getguard}{\mathit{guard}}
\newcommand{\inst}{\mathit{Inst}}
\newcommand{\instspec}{\hat{\mathcal{G}}}

\newcommand{\traceproj}[2]{#1 {\shortdownarrow} {#2}}
\newcommand{\project}{\emph{TraceToWord}}

\begin{abstract}
Several real-world libraries (e.g., reentrant locks, GUI frameworks, serialization libraries) require their clients to use the provided API in a manner that conforms to a context-free specification.  Motivated by this observation,
this paper describes a new technique for verifying the correct usage
of context-free API protocols.  
The key idea underlying our technique is to over-approximate the program's  
feasible API call sequences using a context-free grammar (CFG) and then check language inclusion between  this grammar and the specification. 
However, since this inclusion check may fail
due to imprecision in the program's CFG abstraction, we propose a
novel refinement technique to progressively improve the CFG. In
particular, our method obtains counterexamples from CFG inclusion queries and 
uses them to introduce new non-terminals and productions to the
grammar while still over-approximating the program's relevant
behavior. 

We have implemented the proposed  algorithm in a tool called
\toolname\ and evaluate it on {10} popular Java applications that use at least one API with a context-free specification. Our evaluation shows that \toolname\ is able to verify  correct usage of the API in clients that use it correctly and produces counterexamples for those that do not. We also compare our method against three relevant baselines and demonstrate that \toolname\ enables  verification of safety properties that are beyond the reach of existing tools.

\end{abstract}

\maketitle

\section{Introduction}\label{sec:intro}

Over the last decade, there has been a flurry of research activity on checking the correct usage of APIs~\cite{typestate-fink,bierhoff2009,aldrich2009,joshi2008,arzt2015,pradel2012,lam2004,bierhoff2007}. Despite significant advances in this area, almost all  existing verification techniques focus on \emph{typestate analysis}~\cite{strom1986typestate}, which requires the API protocol to be expressible as a \emph{regular language}. In reality, however, several APIs have context-free --rather than regular-- specifications.  For instance, almost all reentrant lock APIs require   calls to \code{lock} to be balanced by a corresponding call to \code{unlock}.
Similarly, many  APIs provide functionality for saving and restoring internal state, and it is an error to call  \code{restore}  more times than the corresponding \code{save} function. As a final example, in APIs for structured document formats (e.g., JSON), the usage of the library needs to conform  to the underlying context-free document specification. All of these examples are instances of context-free API protocols, and incorrect usage of such APIs typically results in  run-time exceptions or resource leaks.

Motivated by this observation, prior research has developed
\emph{run-time} techniques for specifying context-free properties and
monitoring them during program
execution~\cite{pql,hawk,meredith2010efficient,jin2012javamop}. However,
there has been very little (if any) work on \emph{statically}
verifying conformance between a program and a context-free API
protocol. In this paper, we present a new verification technique that
addresses this problem.  In particular, given a specification
expressed as a \emph{parameterized} context-free grammar (CFG)
$\cfgspec$ and a program $\program$ using that API, our method
automatically checks whether or not $\program$ conforms to protocol
$\cfgspec$. However, solving this problem introduces two key technical
challenges that motivate the novel components of our solution: First,
we need to prove that the program satisfies the API protocol for all,
potentially infinite, relevant objects created by the input
program. To address this challenge, we propose a novel program
instrumentation that transforms the input program so it uses the same
vocabulary as $\cfgspec$ and ensures that if the transformed program
conforms to the API protocol so does the original. Second, because
such APIs are often used in recursive procedures, it is important to
reason precisely both about inter-procedural control flow as well as
feasible API call sequences. Since \emph{both} of these properties,
namely matching call-and-return structure as well as the target API
protocol, are context-free, standard program analysis techniques, such
as CFL reachability~\cite{reps1995precise} or visibly pushdown
automata~\cite{alur2004visibly}, do not address our problem. Instead,
we reduce the context-free protocol verification problem to that of
checking inclusion between two CFGs\footnote{While inclusion checking
  between two CFGs is undecidable, many problems of practical interest
  can be solved by existing tools.}  and propose a
\emph{counterexample-guided abstraction refinement (CEGAR)} approach
for checking whether \emph{every} feasible execution of the program
belongs to the grammar defined by the protocol~(see
Figure~\ref{fig:cegar}).

The heart of our technique consists of a novel abstraction mechanism that represents the input program $\program$ as a context-free grammar $\cfgprog$, whose language $\lang(\cfgprog)$  defines $\program$'s feasible API call sequences. The productions $R$ of this grammar model relevant API calls  as well as intra- and inter-procedural control-flow. For instance, a production  such as $L_1 \rightarrow f L_2$ indicates that API method $f$ is called at program location $L_1$ and  that $L_2$ is a successor of $L_1$. In addition,  productions precisely model inter-procedural control flow  and enforce that every call statement must be matched by its corresponding return.

While the CFG extracted from the program is 
 always \emph{sound}, it may
 be \emph{imprecise} due to data  dependencies that are not captured by the current CFG productions. That
is, if an API call sequence $w$ is feasible in some program execution,
then $w$ is guaranteed to be in $\lang(\cfgprog)$; however, the
membership of $w$ in $\lang(\cfgprog)$ does not  guarantee
the feasibility of the corresponding API call sequence. Our verification approach deals with
this potential imprecision by using a novel abstraction refinement technique that iteratively improves
the program's CFG abstraction until the property can be either refuted or verified.

In more detail, our approach works as follows: First, given  context-free protocol $\cfgspec$ and  current program abstraction $\cfgprog$, we query whether there exists a word $w$ that is in $\cfgprog$ but not $\cfgspec$. If not, then the algorithm terminates with a proof of correctness.  Otherwise, our method reconstructs the corresponding program path $\pi$ associated with $w$ and checks its feasibility using an SMT solver. If  $\pi$ is indeed feasible, then so is the call sequence $w$, and our method terminates with a real counterexample. Otherwise, $w$ must be a \emph{spurious} counterexample caused by imprecision in the CFG. In this case, our algorithm refines the CFG abstraction by computing a proof of infeasibility of $\pi$ in the form of a \emph{nested sequence interpolant}~\cite{nested-interp}. Similar to many other software model checkers, the interpolant drives the refinement process inside the CEGAR loop; however, \emph{unlike}  other  techniques, our approach uses the interpolant to figure out which new non-terminals and productions to add to the grammar. 
In essence, these new non-terminals correspond to ``clones'' of existing program locations and allow us to selectively introduce both intra- and inter-procedural path-sensitivity to our CFG-based program abstraction.

\begin{figure}[!t]
\begin{center}
\includegraphics[scale=0.4]{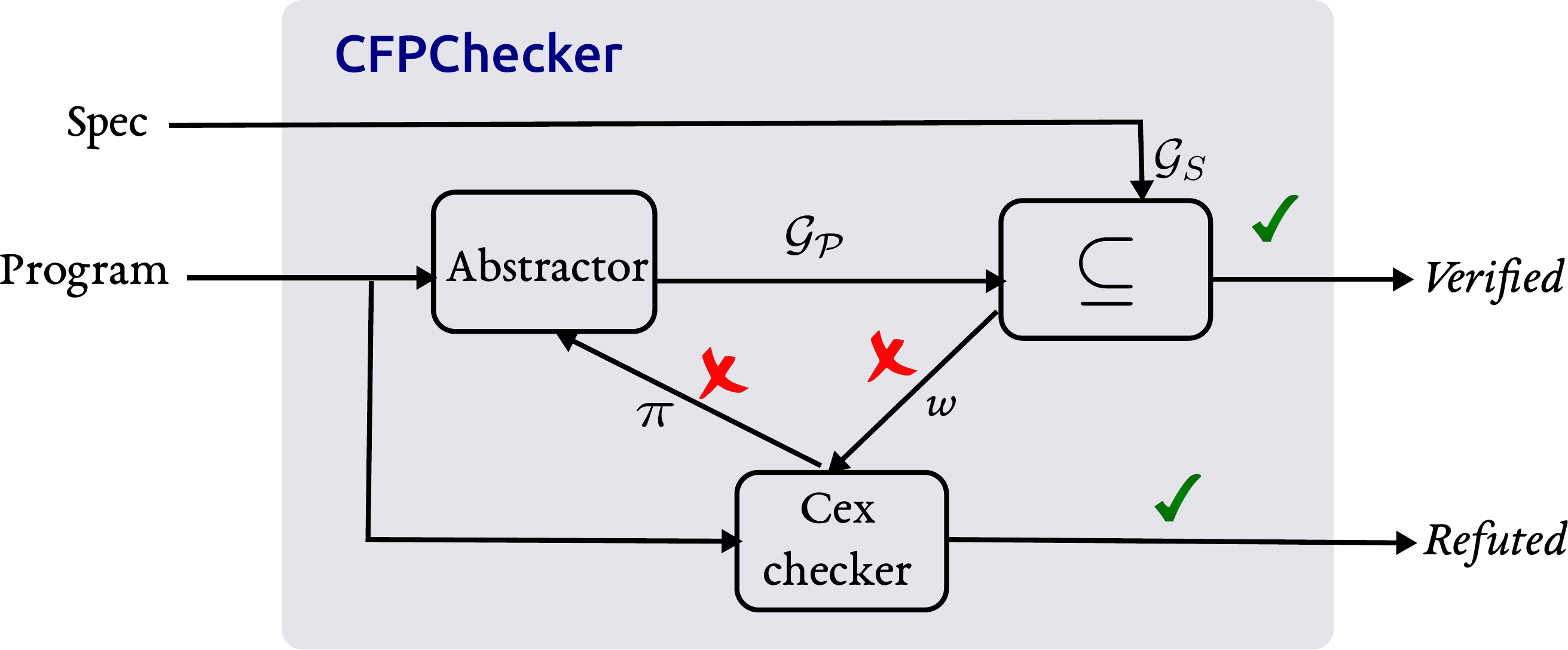}
\caption{Overview of verification approach}\label{fig:cegar}
\end{center}
\end{figure}

We have implemented our proposed verification algorithm in a prototype
called \toolname\ for Java programs and evaluated it on 10 widely-used clients of 5 popular APIs with context-free specifications. Our evaluation demonstrates that \toolname\  is able to verify  correct usage of the API in clients that use it correctly and produces counterexamples for those that do not. We also implement and evaluate three  baselines that reduce the problem to assertion checking and then discharge these assertions using existing tools. Our experiments demonstrate  that \toolname\  is practical enough to successfully analyze real-world Java applications and that it enables  the verification of safety properties that are beyond the reach of existing tools. 

In summary, this paper makes the following contributions:

\begin{itemize}[leftmargin=*]
\item We propose a novel CEGAR-based verification algorithm for
 verifying correct usage of context-free API protocols.
\item We describe a new CFG-based program abstraction that over-approximates
feasible API call sequences.
\item We propose a new refinement method that \emph{selectively and
  modularly} adds path-sensitivity to the program abstraction by
  introducing new non-terminals and productions.
\item We evaluate our method on widely-used clients of popular Java APIs with context-free specifications and demonstrate that our proposed approach is applicable to real-world software verification tasks. 
\end{itemize} 

\section{Motivating Example}\label{sec:motiv}

In this section, we give a high level overview of our approach through a  simple motivating example. Consider a re-entrant lock API that requires every call to \code{lock} on some object \code{o} to be matched by the same number of calls to \code{unlock} on \code{o}. This property is context-free but not regular because it requires "counting" the number of calls to \code{lock} and \code{unlock}. In our framework, the user can specify this property using the following parametrized context-free grammar $\cfgspec$:
\begin{equation}\label{eq:spec}
S \ \rightarrow \ \epsilon \ \ | \ \ \$1.lock() \ S \ \$1.unlock() \ S
\end{equation}
This CFG is parametrized in the sense that it uses a "wildcard"  symbol $\$1$ that matches any object of type \code{Lock}. Thus, the specification requires that, for \emph{every} object $o$, each call \code{o.lock()} must be matched by a  call to \code{o.unlock()}.

\begin{figure}
\begin{subfigure}{.5\textwidth}
\centering
\begin{varwidth}{.5\textwidth}
\begin{Verbatim}[commandchars=\\\{\},numbers=left,firstnumber=1,stepnumber=1]
\PY{k+kt}{void} \PY{n+nf}{foo}\PY{o}{(}\PY{n}{Lock} \PY{n}{l}\PY{o}{)}\PY{o}{\PYZob{}}
  \PY{k}{if} \PY{o}{(}\PY{o}{*}\PY{o}{)} \PY{o}{\PYZob{}}
    \PY{n}{acquire}\PY{o}{(}\PY{n}{l}\PY{o}{)}\PY{o}{;}
    \PY{n}{foo}\PY{o}{(}\PY{n}{l}\PY{o}{)}\PY{o}{;}
    \PY{n}{release}\PY{o}{(}\PY{n}{l}\PY{o}{)}\PY{o}{;}
  \PY{o}{\PYZcb{}}
\PY{o}{\PYZcb{}}

\PY{k+kt}{void} \PY{n+nf}{acquire}\PY{o}{(}\PY{n}{Lock} \PY{n}{l1}\PY{o}{)}\PY{o}{\PYZob{}}
  \PY{n}{l1}\PY{o}{.}\PY{n+na}{lock}\PY{o}{(}\PY{o}{)}\PY{o}{;}
\PY{o}{\PYZcb{}}

\PY{k+kt}{void} \PY{n+nf}{release}\PY{o}{(}\PY{n}{Lock} \PY{n}{l2}\PY{o}{)}\PY{o}{\PYZob{}}
  \PY{n}{l2}\PY{o}{.}\PY{n+na}{unlock}\PY{o}{(}\PY{o}{)}\PY{o}{;}
\PY{o}{\PYZcb{}}
\end{Verbatim}
\end{varwidth}
\caption{Original Program}
\label{fig:origprog}
\end{subfigure}%
\begin{subfigure}{.5\textwidth}
\centering
\begin{varwidth}{.5\textwidth}
\begin{Verbatim}[commandchars=\\\{\},numbers=left,firstnumber=1,stepnumber=1]
\PY{k+kd}{static} \PY{n}{Lock} \PY{n}{\PYZdl{}1} \PY{o}{=} \PY{o}{*}\PY{o}{;}

\PY{k+kt}{void} \PY{n+nf}{foo}\PY{o}{(}\PY{n}{Lock} \PY{n}{l}\PY{o}{)}\PY{o}{\PYZob{}}
  \PY{k}{if} \PY{o}{(}\PY{o}{*}\PY{o}{)} \PY{o}{\PYZob{}}
    \PY{n}{acquire}\PY{o}{(}\PY{n}{l}\PY{o}{)}\PY{o}{;}
    \PY{n}{foo}\PY{o}{(}\PY{n}{l}\PY{o}{)}\PY{o}{;}
    \PY{n}{release}\PY{o}{(}\PY{n}{l}\PY{o}{)}\PY{o}{;}
  \PY{o}{\PYZcb{}}
\PY{o}{\PYZcb{}}

\PY{k+kt}{void} \PY{n+nf}{acquire}\PY{o}{(}\PY{n}{Lock} \PY{n}{l1}\PY{o}{)}\PY{o}{\PYZob{}}
  \PY{k}{if} \PY{o}{(}\PY{n}{l1} \PY{o}{=}\PY{o}{=} \PY{n}{\PYZdl{}1}\PY{o}{)}
    \PY{n}{\PYZdl{}1}\PY{o}{.}\PY{n+na}{lock}\PY{o}{(}\PY{o}{)}\PY{o}{;}
\PY{o}{\PYZcb{}}

\PY{k+kt}{void} \PY{n+nf}{release}\PY{o}{(}\PY{n}{Lock} \PY{n}{l2}\PY{o}{)}\PY{o}{\PYZob{}}
  \PY{k}{if} \PY{o}{(}\PY{n}{l2} \PY{o}{=}\PY{o}{=} \PY{n}{\PYZdl{}1}\PY{o}{)}
    \PY{n}{\PYZdl{}1}\PY{o}{.}\PY{n+na}{unlock}\PY{o}{(}\PY{o}{)}\PY{o}{;}
\PY{o}{\PYZcb{}}
\end{Verbatim}

\end{varwidth}
\caption{Transformed Program}
\label{fig:transprog}
\end{subfigure}
\caption{Motivating Example}
\label{fig:motiv}
\end{figure}

\begin{figure}
  \centering
  \begin{subfigure}{\textwidth}
    \centering
    \includegraphics[width=\textwidth]{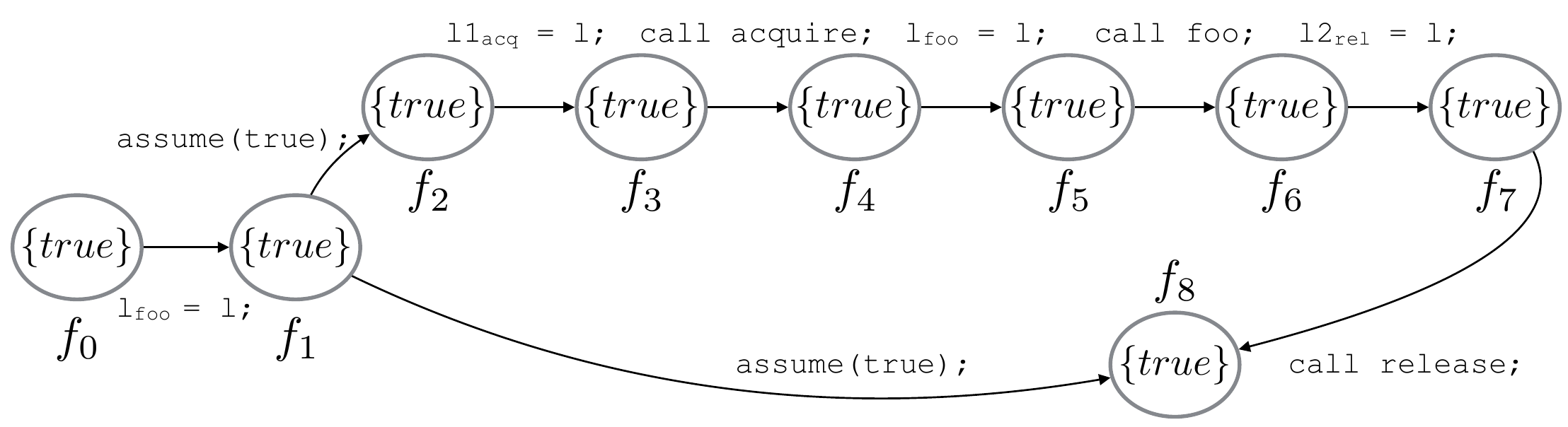}
    \caption{Initial $\mathit{PCFA}$ for foo.}
  \end{subfigure}
  \par\bigskip
  \begin{subfigure}{\textwidth}
    \centering
    \begin{subfigure}{.5\textwidth}
      \includegraphics[width=\textwidth]{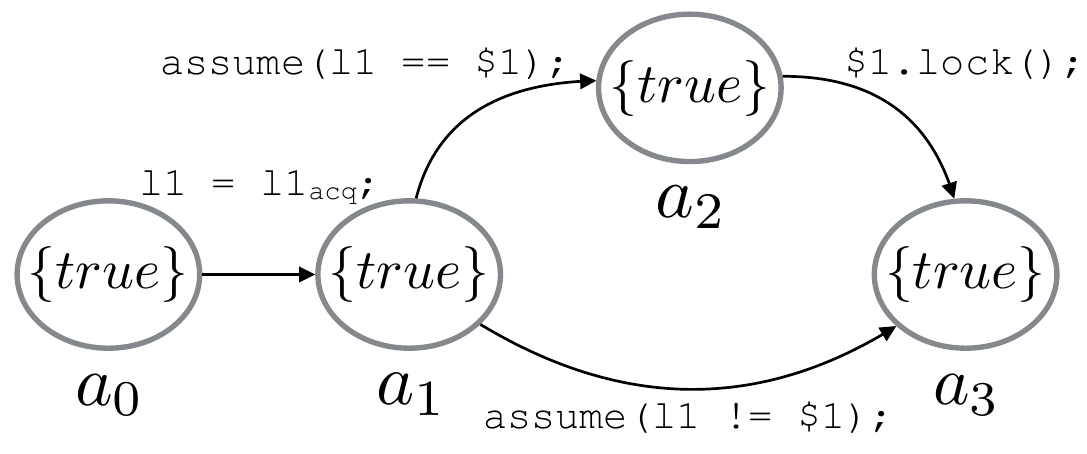}
      \caption{Initial $\mathit{PCFA}$ for acquire.}\label{fig:init-acquire}
    \end{subfigure}%
    \begin{subfigure}{.5\textwidth}
      \includegraphics[width=\textwidth]{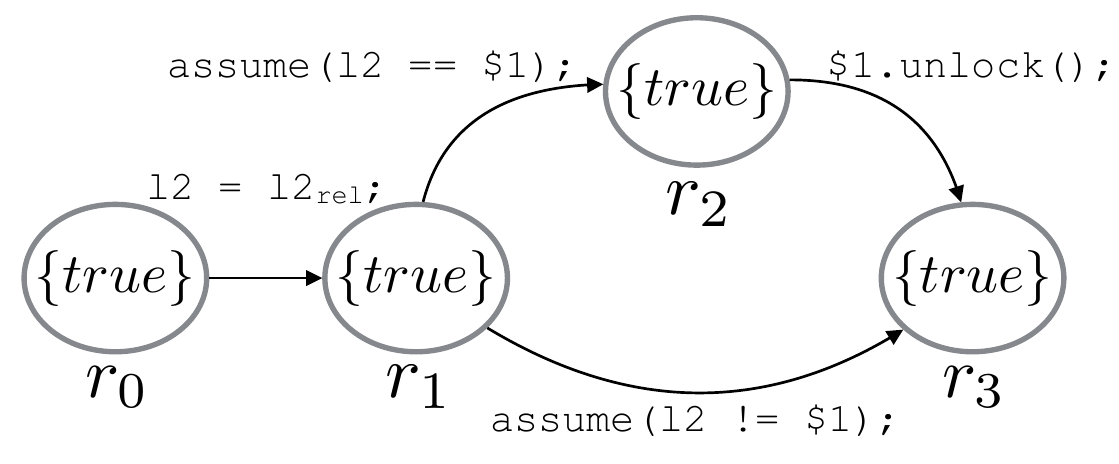}
      \caption{Initial $\mathit{PCFA}$ for release.}
    \end{subfigure}
  \end{subfigure}%
  \caption{Initial $\mathit{PCFAs}$ for input program. The PCFAs contain additional formal-to-actual assignments.}
  \label{fig:pcfaorig}
\end{figure}

\begin{figure}[h]
  \centering
  \begin{subfigure}{.3\textwidth}
    \centering
    \begin{grammar}
      \let\syntleft\relax
      \let\syntright\relax
      <$\mathit{Foo}$> $\rightarrow$ $\mathcal{F}_{0}$

      <$\mathcal{F}_{0}$> $\rightarrow$ $\mathcal{F}_1$

      <$\mathcal{F}_1$> $\rightarrow$ $\mathcal{F}_{2}$ | $\mathcal{F}_8$

      <$\mathcal{F}_{2}$> $\rightarrow$ $\mathcal{F}_{3}$

      <$\mathcal{F}_{3}$> $\rightarrow$ $\mathit{Acquire}$ $\mathcal{F}_{4}$

      <$\mathcal{F}_{4}$> $\rightarrow$ $\mathcal{F}_{5}$

      <$\mathcal{F}_{5}$> $\rightarrow$ $\mathit{Foo}$ $\mathcal{F}_{6}$

      <$\mathcal{F}_{6}$> $\rightarrow$ $\mathcal{F}_{7}$

      <$\mathcal{F}_{7}$> $\rightarrow$ $\mathit{Release}$ $\mathcal{F}_{8}$

      <$\mathcal{F}_{8}$> $\rightarrow$ $\epsilon$
    \end{grammar}
  \end{subfigure}%
  \begin{subfigure}{.3\textwidth}
    \centering
    \begin{grammar}
      \let\syntleft\relax
      \let\syntright\relax
      <$\mathit{Acquire}$> $\rightarrow$ $\mathcal{A}_0$

      <$\mathcal{A}_0$> $\rightarrow$ $\mathcal{A}_1$

      <$\mathcal{A}_1$> $\rightarrow$ $\mathcal{A}_2$ | $\mathcal{A}_3$

      <$\mathcal{A}_2$> $\rightarrow$ \code{\$1.lock()} $\mathcal{A}_3$

      <$\mathcal{A}_3$> $\rightarrow$ $\epsilon$ \\
    \end{grammar}
  \end{subfigure}%
  \begin{subfigure}{.3\textwidth}
    \centering
    \begin{grammar}
      \let\syntleft\relax
      \let\syntright\relax
      <$\mathit{Release}$> $\rightarrow$ $\mathcal{R}_0$

      <$\mathcal{R}_0$> $\rightarrow$ $\mathcal{R}_1$

      <$\mathcal{R}_1$> $\rightarrow$ $\mathcal{R}_2$ | $\mathcal{R}_3$

      <$\mathcal{R}_2$> $\rightarrow$ \code{\$1.unlock()} $\mathcal{R}_3$

      <$\mathcal{R}_3$> $\rightarrow$ $\epsilon$
    \end{grammar}
  \end{subfigure}
  \par\bigskip
  \hrulefill
  \caption{Initial context-free grammar.}
  \label{fig:initcfg}
\end{figure}

To illustrate  our technique,  Figure~\ref{fig:motiv}(a) shows a very simple client of this \code{Lock} API. Here, \code{foo} is a recursive procedure that calls \code{l.lock} before every recursive call to \code{foo} and calls \code{l.unlock} afterwards. Since the receiver object is the same before and after the call, the specification from Equation~\ref{eq:spec} is satisfied. In the remainder of this section, we explain how our technique  verifies correct usage of the \code{Lock} API in this example.

The first step in our technique is to automatically instrument the program from Figure~\ref{fig:motiv}(a) so that   API calls in the program involve the same wildcard symbol $\$1$ used in the specification. The instrumented version is shown in Figure~\ref{fig:motiv}(b), which uses a new global variable called \code{\$1} (i.e., the wildcard symbol in the grammar) and 
replaces every call to \code{x.lock()} (resp. \code{x.unlock()}) with the conditional  invocation \code{if(x = \$1) \$1.lock()} (resp. \code{if(x = \$1) \$1.unlock()}).   Intuitively, the goal of this instrumentation is two-fold: First, it  ensures that the CFG abstraction of the program uses the same "vocabulary" (i.e., terminals) as the specification CFG. Second, it deals with challenges that arise from potential aliasing between pointers.

In the next step,  our method extracts a
context-free grammar that over-approximates the relevant API call
behavior of the program. Towards this goal, we represent the program
as a mapping from each function to a \emph{predicated control-flow
  automaton} (\emph{PCFA}) that will be iteratively refined as the
algorithm progresses. {At a high level, a {PCFA} 
captures control-flow within a method while also maintaining a mapping from
  program locations to a set of logical predicates.}  For example,
Figure~\ref{fig:pcfaorig} shows the initial PCFAs for
Figure~\ref{fig:motiv}(b): here, nodes correspond to program locations, and edges correspond to 
transitions.
Observe that the PCFAs from Figure~\ref{fig:pcfaorig} contain a \emph{single} node for each program location;  hence, these PCFAs look like  standard \emph{control flow automata (CFA)} used in software model checking~\cite{reps1995precise,nested-interp}.
 However,  the PCFA representation diverges from a standard CFA   as the algorithm proceeds. In particular, the PCFA can  contain multiple nodes for the {same} program location  and allows our method to selectively introduce path-sensitivity to the program abstraction.

\begin{figure}
  \begin{subfigure}{.5\textwidth}
    \includegraphics[width=.9\textwidth]{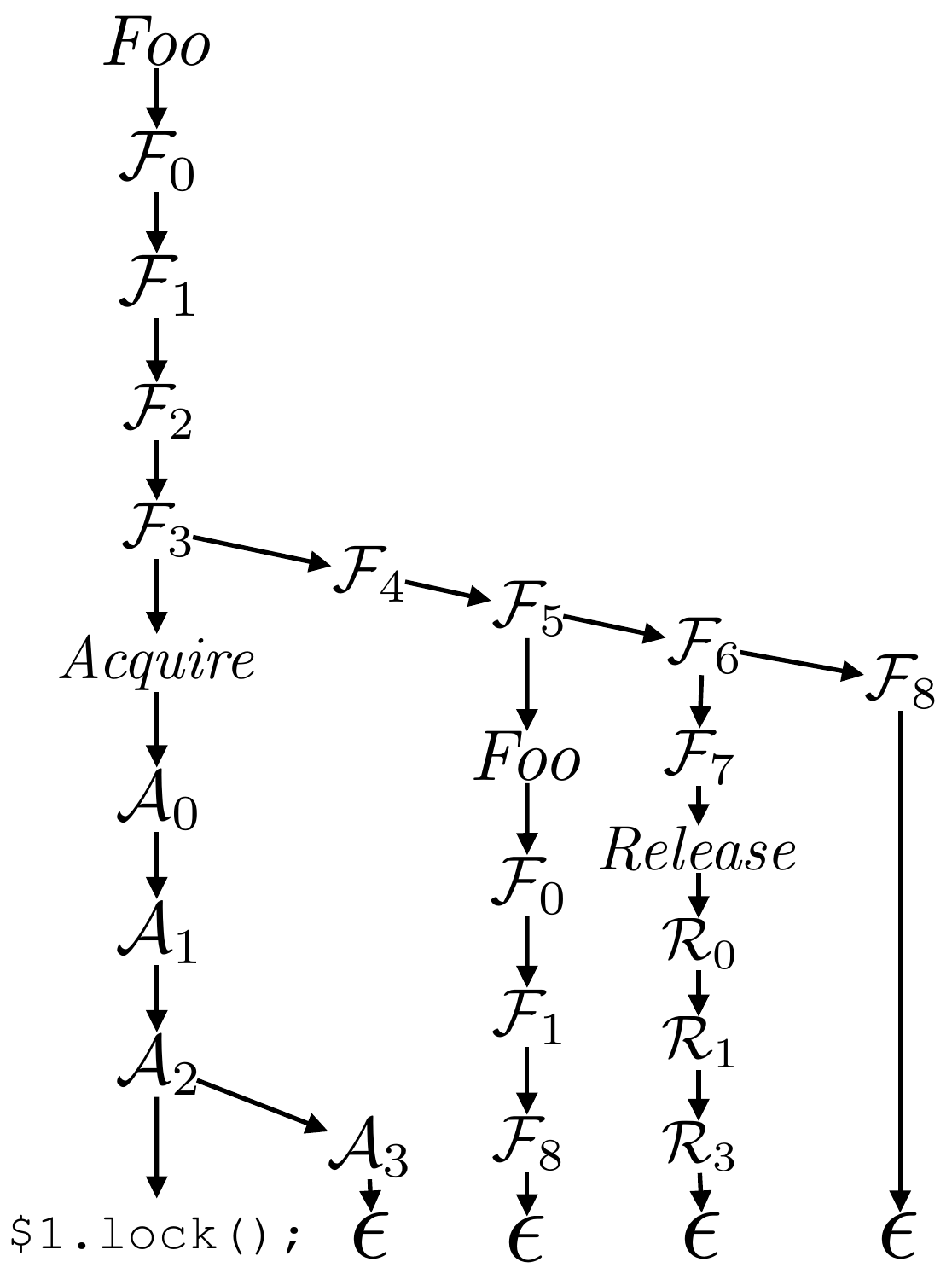}
    \caption{Parse Tree.}\label{fig:cex-tree}
  \end{subfigure}%
  \begin{subfigure}{.5\textwidth}
    \includegraphics[width=.9\textwidth]{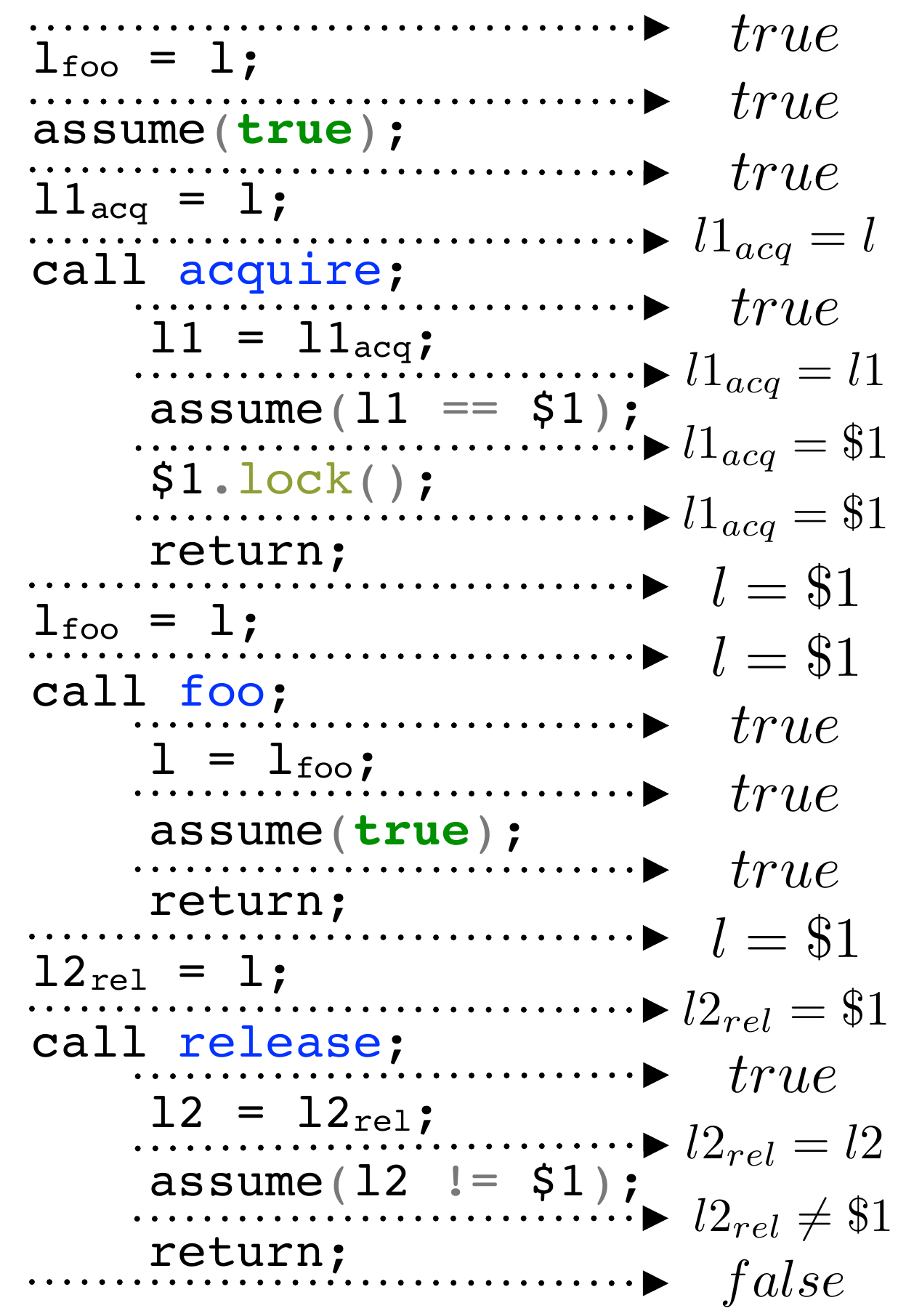}
    \caption{Trace \& Interpolants.}\label{fig:cex-trace}
  \end{subfigure}
  \caption{Tree and Trace for Counterexample \code{\$1.lock()}.}\label{fig:cex}
\end{figure}

Given these initial PCFAs, our method programmatically
extracts  from them a context-free grammar over-approximating the program's feasible API call sequences.
 In particular, Figure~\ref{fig:initcfg} shows 
the initial CFG abstraction for our  example.
Here, non-terminals  (e.g., $\mathcal{F}_1,
\mathcal{A}_2$)  correspond to nodes (e.g., $f_1, a_2$) in the PCFAs, and  terminals
(e.g., \code{\$1.lock()}) denote API calls. Additionally, there
is one non-terminal symbol (e.g., \emph{Foo, Acquire}) for each
method. The productions in the CFG are obtained
directly from the PCFA by ignoring all statements that are not
function calls: For example, the production $\mathcal{A}_2
\ \rightarrow \ \code{\$1.lock()} \ \mathcal{A}_3$ comes from the PCFA
edge from $a_2$ to $a_3$.
In addition, the CFG productions  faithfully and precisely 
model
inter-procedural control flow. For instance, the production $\mathcal{F}_3
\ \rightarrow \ Acquire \ \mathcal{F}_4$  models the call from $Foo$ to $Acquire$ and 
  $\mathcal{A}_3\ \rightarrow \ \epsilon$ models its corresponding return.

\begin{figure}
  \centering
  \begin{subfigure}{\textwidth}
    \centering
    \includegraphics[width=\textwidth]{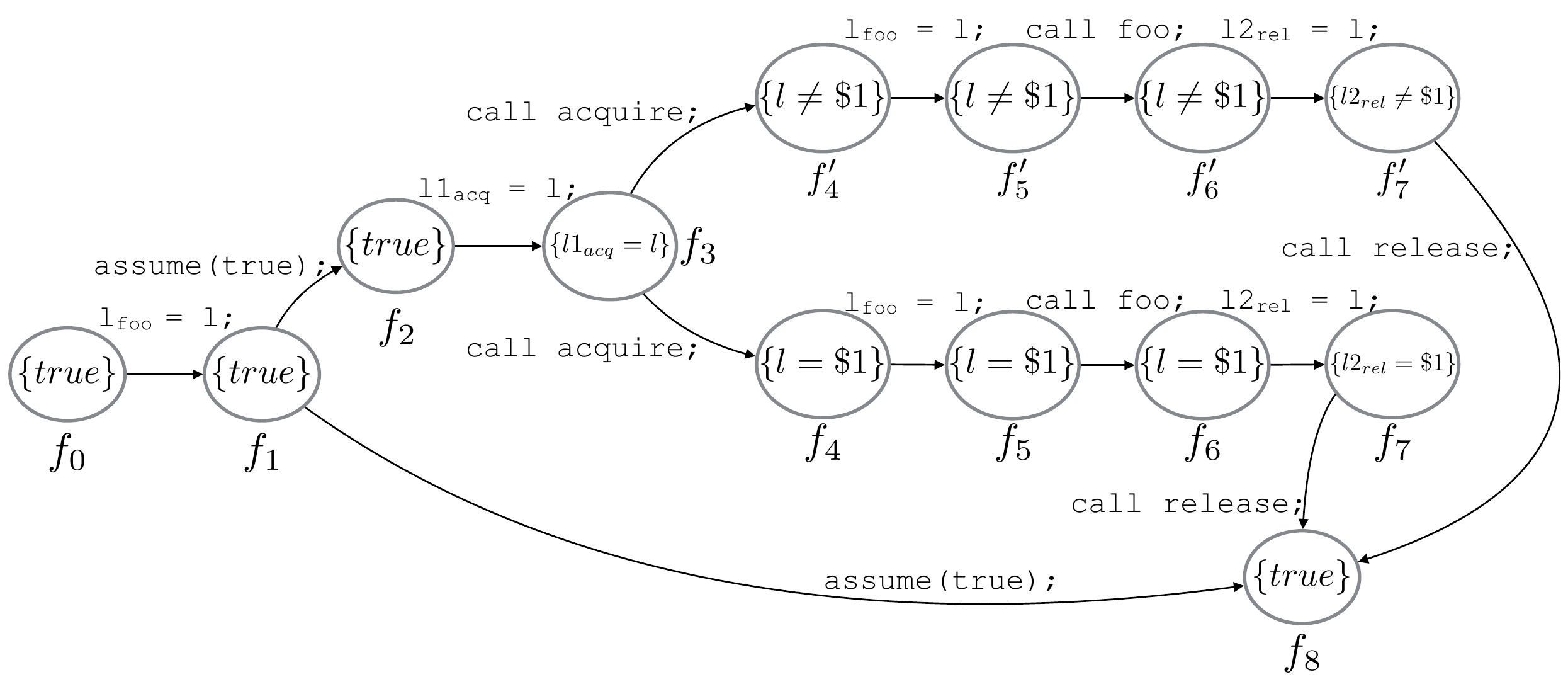}
    \caption{Refined $\mathit{PCFA}$ for foo.}
  \end{subfigure}
  \par\bigskip
  \begin{subfigure}{\textwidth}
    \centering
    \begin{subfigure}{.5\textwidth}
      \includegraphics[width=\textwidth]{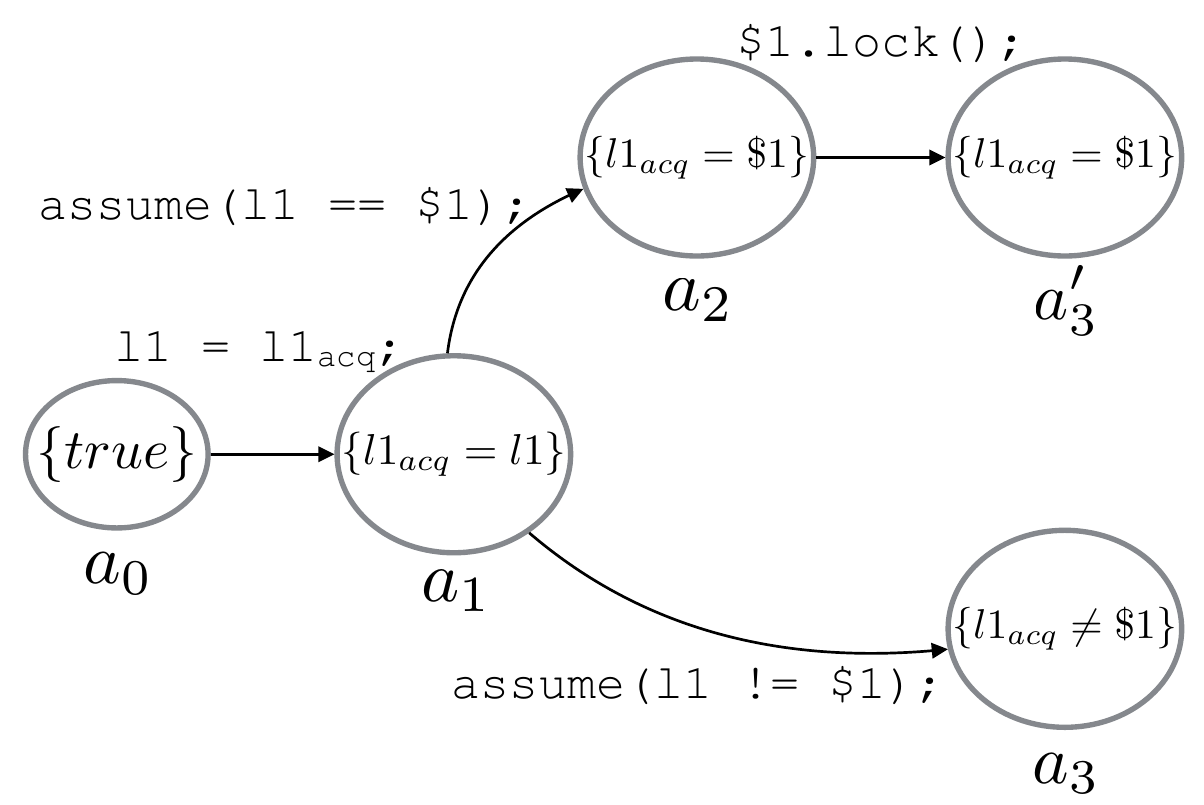}
      \caption{Refined $\mathit{PCFA}$ for acquire.}\label{fig:acq-ref}
    \end{subfigure}%
    \begin{subfigure}{.5\textwidth}
      \includegraphics[width=\textwidth]{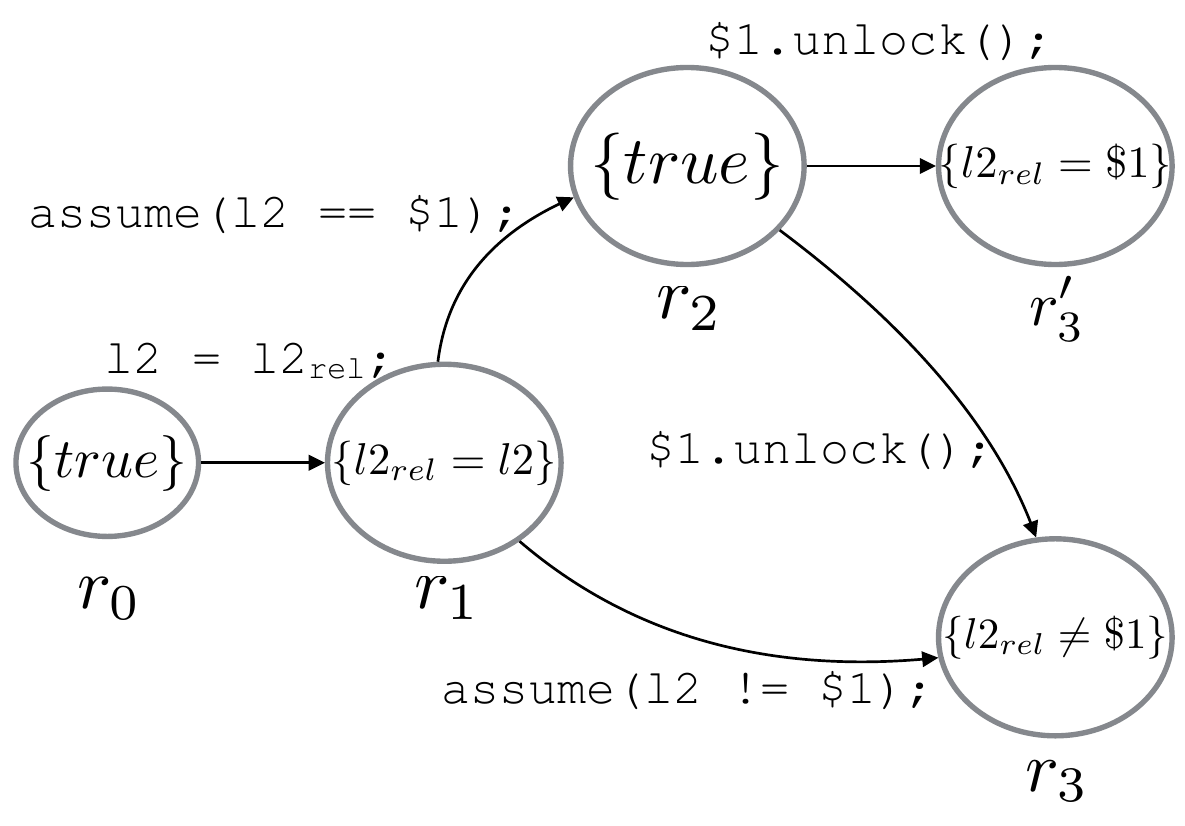}
      \caption{Refined $\mathit{PCFA}$ for release.}\label{fig:rel-ref}
    \end{subfigure}
  \end{subfigure}%
  \caption{Refined $\mathit{PCFAs}$ for input program.}
  \label{fig:pcfaref}
\end{figure}

Next,  our method
checks inclusion between the grammar $\cfgprog$ extracted from the program and API protocol
$\cfgspec$. While this problem is, in general, undecidable, we have found the resulting CFG inclusion checking problems to be amenable to automation by modern tools.  Going back to our running example, the language of $\cfgprog$ from
Figure~\ref{fig:initcfg} is \emph{not} a subset of the language of
$\cfgspec$ --- for example, the word \code{\$1.lock()} can be generated using
$\cfgprog$ but not $\cfgspec$. {This means that either the
  program actually misuses the API or the current abstraction is imprecise. In
  order to determine which one, our method
  maps the word \code{\$1.lock()} to an execution path of the program. Towards this goal, we
  first obtain the parse tree from   Figure~\ref{fig:cex-tree}
  that shows how  \code{\$1.lock()}}  can be derived from $\cfgprog$. This derivation corresponds precisely to the program path,  shown in Figure~\ref{fig:cex-trace}. {Furthermore,
  observe that this path goes through the ``then'' branch of the
  \code{if} statement in method \code{acquire} and the ``else'' branch
  in method \code{release}. However, this path is clearly infeasible, so 
  we need to refine  $\cfgprog$ to eliminate
the spurious derivation.

Our method refines the program's CFG abstraction by  adding new non-terminals and productions to the grammar. Towards this goal, we first refine the PCFA abstraction by selectively cloning some program locations, with the goal of introducing path-sensitivity where needed. The cloning of PCFA nodes is driven by an interpolation engine  that computes a \emph{sequence of nested interpolants}~\cite{nested-interp}. In particular, the right-hand side of Figure~\ref{fig:cex-trace} shows the interpolants computed for each program location for our running example. Intuitively, "tracking" these predicates at the corresponding program location would allow us to remove the spurious  trace. Thus, in the next iteration, we generate the new PCFAs shown in Figure~\ref{fig:pcfaref} by cloning all PCFA nodes that correspond to program locations in the counterexample. Observe that the refined PCFAs  contain multiple nodes (e.g., $f_4, f_4'$) for the same program location, and the predicates in the PCFA correspond to those that appear in the interpolant. For instance, even though  nodes $r_3, r_3'$ both represent the same program location, one is annotated with predicate $l2_{rel} \neq \$1$, whereas $r_3'$ is annotated with  $l2_{rel} = \$1$. Furthermore, the refined PCFA contains an edge between two nodes iff the semantics of the statement labeling that edge are consistent with the annotations of the source and target nodes. For instance, there is an edge from node $a_1$ to $a_3$ but not from $a_1$ to $a_3'$ because the predicates $l1_{acq} = l1$, $l1_{acq} = \$1$ labeling $a_1$ and $a_3'$  are inconsistent with the statement \code{assume(l1 != \$1)}.

\begin{figure}
  \centering
  \begin{subfigure}{.3\textwidth}
    \centering
    \begin{grammar}
      \let\syntleft\relax
      \let\syntright\relax
      <$\mathit{Foo}$> $\rightarrow$ $\mathcal{F}_{0}$

      <$\mathcal{F}_{0}$> $\rightarrow$ $\mathcal{F}_1$

      <$\mathcal{F}_1$> $\rightarrow$ $\mathcal{F}_{2}$ | $\mathcal{F}_{8}$

      <$\mathcal{F}_{2}$> $\rightarrow$ $\mathcal{F}_3$

      <$\mathcal{F}_{3}$> $\rightarrow$ $\mathit{Acquire}_{\phi_{1}}$ $\mathcal{F}_{4}$
      \alt $\mathit{Acquire}_{\phi_{2}}$ $\mathcal{F}_{4}'$

      <$\mathcal{F}_{4}$> $\rightarrow$ $\mathcal{F}_5$

      <$\mathcal{F}_{4}'$> $\rightarrow$ $\mathcal{F}_5'$

      <$\mathcal{F}_{5}$> $\rightarrow$ $\mathit{Foo}$ $\mathcal{F}_{6}$

      <$\mathcal{F}_{5}'$> $\rightarrow$ $\mathit{Foo}$ $\mathcal{F}_{6}'$

      <$\mathcal{F}_{6}$> $\rightarrow$ $\mathcal{F}_7$

      <$\mathcal{F}_{6}'$> $\rightarrow$ $\mathcal{F}_7'$

      <$\mathcal{F}_{7}$> $\rightarrow$ $\mathit{Release}_{\phi_3}$ $\mathcal{F}_{8}$

      <$\mathcal{F}_{7}'$> $\rightarrow$ $\mathit{Rlease}_{\phi_4}$ $\mathcal{F}_{8}$

      <$\mathcal{F}_{8}$> $\rightarrow$ $\epsilon$
    \end{grammar}
  \end{subfigure}%
  \begin{subfigure}{.3\textwidth}
    \centering
    \begin{grammar}
      \let\syntleft\relax
      \let\syntright\relax
      <$\mathit{Acquire}_{\phi_1}$> $\rightarrow$ $\mathcal{A}_{0,\phi_1}$

      <$\mathcal{A}_{0,\phi_1}$> $\rightarrow$ $\mathcal{A}_{1,\phi_1}$

      <$\mathcal{A}_{1,\phi_1}$> $\rightarrow$ $\mathcal{A}_{2,\phi_1}$

      <$\mathcal{A}_{2,\phi_1}$> $\rightarrow$ \code{\$1.lock()} $\mathcal{A}_{3,\phi_1}'$

      <$\mathcal{A}_{3,\phi_1}'$> $\rightarrow$ $\epsilon$\\\\

      <$\mathit{Acquire}_{\phi_2}$> $\rightarrow$ $\mathcal{A}_{0,\phi_2}$

      <$\mathcal{A}_{0,\phi_2}$> $\rightarrow$ $\mathcal{A}_{1,\phi_2}$

      <$\mathcal{A}_{1,\phi_2}$> $\rightarrow$ $\mathcal{A}_{3,\phi_2}$

      <$\mathcal{A}_{3,\phi_2}$> $\rightarrow$ $\epsilon$
    \end{grammar}
  \end{subfigure}%
  \begin{subfigure}{.3\textwidth}
    \centering
    \begin{grammar}
      \let\syntleft\relax
      \let\syntright\relax
      <$\mathit{Release}_{\phi_3}$> $\rightarrow$ $\mathcal{R}_{0,\phi_3}$

      <$\mathcal{R}_{0,\phi_3}$> $\rightarrow$ $\mathcal{R}_{1,\phi_3}$

      <$\mathcal{R}_{1,\phi_3}$> $\rightarrow$ $\mathcal{R}_{2,\phi_3}$

      <$\mathcal{R}_{2,\phi_3}$> $\rightarrow$ \code{\$1.unlock()} $\mathcal{R}_{3,\phi_3}'$

      <$\mathcal{R}_{3,\phi_3}'$> $\rightarrow$ $\epsilon$\\\\
  
      <$\mathit{Release}_{\phi_4}$> $\rightarrow$ $\mathcal{R}_{0,\phi_4}$

      <$\mathcal{R}_{0,\phi_4}$> $\rightarrow$ $\mathcal{R}_{1,\phi_4}$

      <$\mathcal{R}_{1,\phi_4}$> $\rightarrow$ $\mathcal{R}_{2,\phi_4}$

      <$\mathcal{R}_{1,\phi_4}$> $\rightarrow$ $\mathcal{R}_{3,\phi_4}$

      <$\mathcal{R}_{2,\phi_4}$> $\rightarrow$ \code{\$1.unlock()} $\mathcal{R}_{3,\phi_4}$

      <$\mathcal{R}_{3,\phi_4}$> $\rightarrow$ $\epsilon$
    \end{grammar}
  \end{subfigure}
  \par\bigskip
  \hrulefill
  \caption{Refined CFG, where $\phi_1 = \{l1_{acq} = \$1\}$, $\phi_2 =
    \{l1_{acq} \neq \$1\}$, $\phi_3 = \{l2_{rel} = \$1\}$, and $\phi_4
    = \{l2_{rel} \neq \$1\}$.}
  \label{fig:refcfg}
\end{figure}

\begin{figure}
  \begin{subfigure}{.5\textwidth}
    \centering
    \includegraphics[width=.8\textwidth]{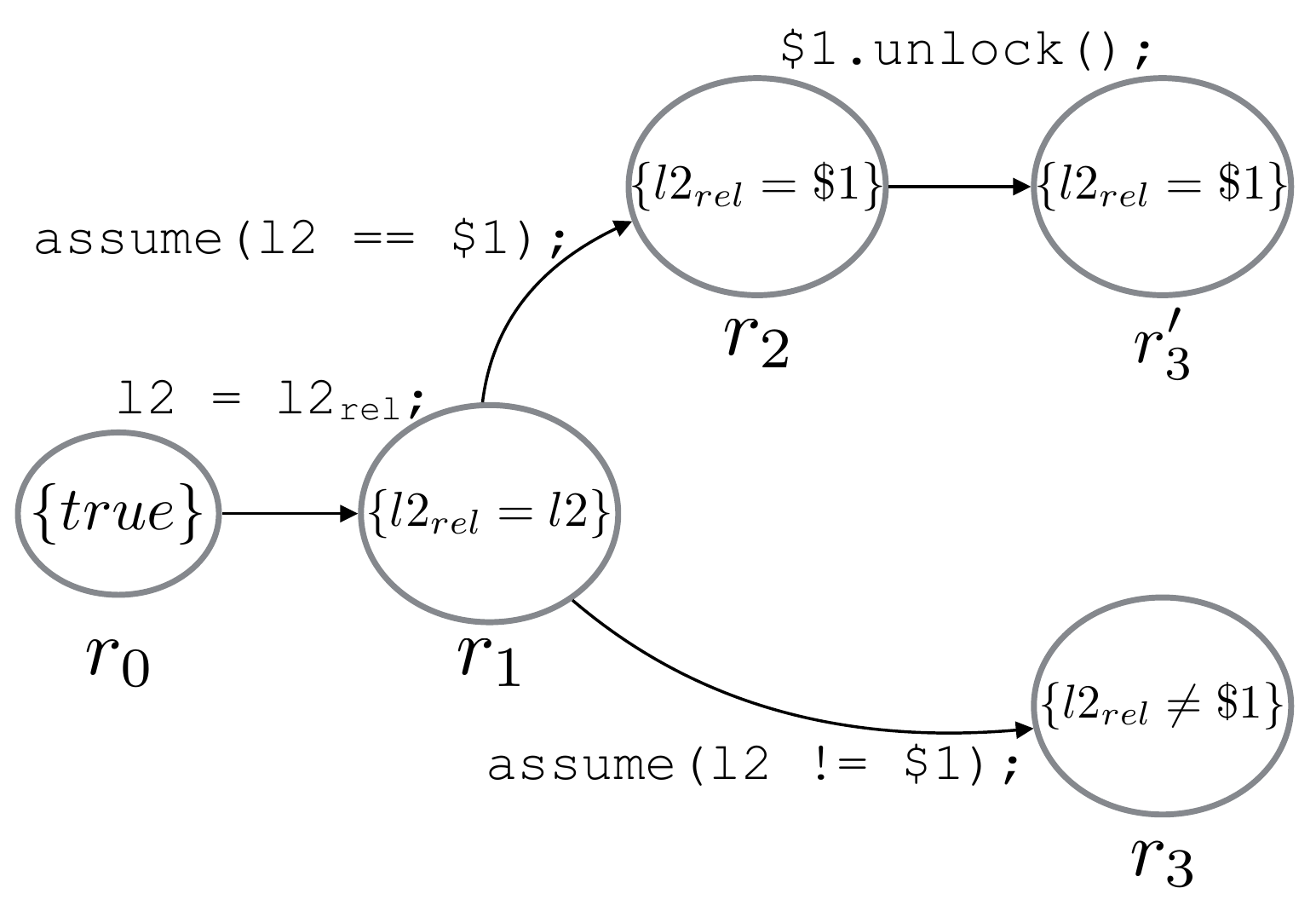}
    \caption{Second refined PCFA for method \code{release}.}\label{fig:rel-ref-pcfa2}
  \end{subfigure}%
  \begin{subfigure}{.5\textwidth}
    \centering
    \begin{grammar}
      \let\syntleft\relax
      \let\syntright\relax
      <$\ \ \ \ \ \ \ \ \ \ \ \ \ \mathit{Release}_{\phi_3}$> $\rightarrow$ $\mathcal{R}_{0,\phi_3}$

      <$\ \ \ \ \ \ \ \ \ \ \ \ \ \mathcal{R}_{0,\phi_3}$> $\rightarrow$ $\mathcal{R}_{1,\phi_3}$

      <$\ \ \ \ \ \ \ \ \ \ \ \ \ \mathcal{R}_{1,\phi_3}$> $\rightarrow$ $\mathcal{R}_{2,\phi_3}$

      <$\ \ \ \ \ \ \ \ \ \ \ \ \ \mathcal{R}_{2,\phi_3}$> $\rightarrow$ \code{\$1.unlock()} $\mathcal{R}_{3,\phi_3}'$

      <$\ \ \ \ \ \ \ \ \ \ \ \ \ \mathcal{R}_{3,\phi_3}'$> $\rightarrow$ $\epsilon$\\\\

      <$\ \ \ \ \ \ \ \ \ \ \ \ \ \mathit{Release}_{\phi_4}$> $\rightarrow$ $\mathcal{R}_{0,\phi_4}$

      <$\ \ \ \ \ \ \ \ \ \ \ \ \ \mathcal{R}_{0,\phi_4}$> $\rightarrow$ $\mathcal{R}_{1,\phi_4}$

      <$\ \ \ \ \ \ \ \ \ \ \ \ \ \mathcal{R}_{1,\phi_4}$> $\rightarrow$ $\mathcal{R}_{3,\phi_4}$

      <$\ \ \ \ \ \ \ \ \ \ \ \ \ \mathcal{R}_{3,\phi_4}$> $\rightarrow$ $\epsilon$
    \end{grammar}
    \caption{Refined grammar for $\mathit{Release}$.}\label{fig:rel-ref-cfg2}
  \end{subfigure}
  \caption{Refined PCFA and grammar for method \code{release} (second iteration).}\label{fig:rel-ref-it2}
\end{figure}

Given this new PCFA representation, our verification algorithm
extracts the refined grammar $\cfgprog'$  shown in Figure~\ref{fig:refcfg}.
As before, we construct the grammar
based on PCFA edges; however,
note that there are  two different sets of grammar rules for each of
the methods \code{acquire} and \code{release}. In general, for a given
function $f$, our technique introduces as many non-terminals for $f$ as there are
PCFA nodes for $f$'s exit location. This strategy allows our verification 
algorithm to  lazily perform "method cloning", thereby introducing \emph{inter-procedural} path-sensitivity where needed.
 For instance, observe that there are two
non-terminals ($\mathit{Acquire}_{\phi_1}$,
$\mathit{Acquire}_{\phi_2}$) representing \code{acquire} in
Figure~\ref{fig:refcfg}, and predicates $\phi_1, \phi_2$ correspond to
the predicates $l1_{acq} = \$1$, $l1_{acq} \neq \$1$ labeling nodes
$a_3$ and $a_3'$ in the PCFA from
Figure~\ref{fig:pcfaref}(b). Furthermore, observe that there are two
different sets of grammars for $\mathit{Acquire}_{\phi_1}$ and
$\mathit{Acquire}_{\phi_2}$, and each grammar is generated by looking
at the portion of the PCFA that is backwards reachable from the
corresponding exit node. For example, there is no production
$\mathcal{A}_{1, \phi_2} \rightarrow \mathcal{A}_{2, \phi_2}$ in
Figure~\ref{fig:refcfg} because node $a_2$ is not backwards reachable
from the exit node labeled with $\phi_2$ in
Figure~\ref{fig:pcfaref}(b).

In the second iteration, our algorithm again checks inclusion between
the two grammars, namely $\cfgprog'$ and $\cfgspec$. This time,
$\cfgprog'$ is still not contained in $\cfgspec$, and the new
counterexample is \code{\$1.unlock()}, whose derivation corresponds to a program path that
goes through the ``else'' branch in \code{acquire} and ``then'' branch
in \code{release}.  In this case, the culprit is the PCFA edge between
nodes $r_2$ and $r_3$ in method \code{release}
(Figure~\ref{fig:rel-ref}), which can again be eliminated by
computing nested interpolants and cloning node $r_2$.

In the next and final iteration, our algorithm can now prove that the
language defined by the program's CFG is indeed a subset of the
specification $\cfgspec$, and the algorithm terminates with a proof of
correctness. The final abstraction is identical to the one from previous
iteration except for method \code{release} whose final
PCFA and context-free grammar are shown in
Figure~\ref{fig:rel-ref-it2}.

\section{Problem Statement}\label{sec:probstmt}

In this section, we introduce context-free API protocols and formally define our problem in the context of a simple object-oriented programming language. 

\subsection{Input Language}\label{sec:in-lang}

\begin{figure}
\small
  \centering
    \begin{grammar}
      \let\syntleft\relax
      \let\syntright\relax
      <Class $\mathit{C}$> $::=$ \code{class $\mathit{C}$ \{ }$\mathit{fld}\text{*}$ $\mathit{m}\text{*}\code{ \}}$

      <Field $\mathit{fld}$> $::=$ $f : \uptau = e$ | $\code{static}\ f : \uptau = e$

      <Method $\mathit{m}$> $::=$ $\code{void } \mathit{m}(\vec{v})\code { \{}  s\text{*}; \code{\}}$

      <Stmt $s$> $::=$ $\mathit{skip}$ | $s_1;s_2$ | $v := e$ | $v.f := e$ | \code{assume($p$)} 
     | 
      $\code{if } (p) \code{ \{} s_1 \code{\} else \{} s_2 \code{\}}$ | $v := \code{new}\ C$
      \alt
       $\code{call}\ v.m(\vec{v})$ | $\code{api\_call}\ v.m(\vec{v})$

      <Expr $e$> $::=$ $v$ | $v.f$ | $c$ | $*$ | $e_1 \ominus e_2,$ $\ominus \in \{+,-,\times\}$
      
      <Pred $p$> $::=$ $e$ | $\neg p$ | $p_1\land p_2$ | $p_1 \lor p_2 \ | \ e_1 \oplus e_2$, \  $\oplus \in \{\textless, \textgreater, =\}$
    \end{grammar}
    \caption{Input Language.}
    \label{fig:src-lang}
\end{figure}

Figure~\ref{fig:src-lang} presents the programming language
used for our formalization. In this language, a class consists of 
a set of field declarations followed by a set of method definitions. Fields can be either
object-specific (declared as $f:\tau$) or \code{static}, meaning 
they are shared between all instances of the class.  Statements 
include standard constructs like assignment, load, store,
etc. We differentiate between two
kinds of call statements, namely \code{call} which is a call to a
regular method defined in the same program and \code{api_call} which
invokes a method defined by a third-party API. We assume that the source code of third-party libraries are not available for analysis; thus, we require any side effects of  API calls to be modeled using stub methods. In particular, we assume that each call to an API method \code{foo} in the original program has been replaced by a stub  \code{foo_stub} that invokes \code{foo} and captures its side effects via  assignment. \emph{Thus, in the remainder of the paper, we assume, without loss of generality, that API calls have no side effects on program state. }

For the purposes of this paper, a program state $\pstate$ is a mapping
from program variables ($\vars$) and field references ($\vars \times
\fields$) to an integer value. We use the notation $\langle s, \pstate
\rangle \Downarrow \pstate'$ to indicate that $\pstate'$ is the
resulting state after executing statement $s$ on program state
$\pstate$. Furthermore, we use $sp(s, P)$ to denote the strongest
postcondition of statement $s$ with respect to the first-order logic
formula $P$. A program trace, $\ptrace = \langle s_1, \pstate_1
\rangle, \langle s_2, \pstate_2 \rangle, ..., \langle s_n, \pstate_n
\rangle$, is a sequence of (statement, program state) pairs such that
$\langle s_i, \pstate_i \rangle \Downarrow
\pstate_{i+1}$.\footnote{{We assume that program traces are in SSA
    form. That is, each re-definition of a program variable is
    assigned a unique name within the trace.}}  Given a program
$\program$, we write $\emph{Traces}(\program)$ to denote the
(infinite) set of traces that can arise during executions of
$\program$.

\subsection{Context-Free API Protocols}\label{sec:spec}

We express API protocols using
a (parametrized) context-free grammar $\cfgspec = (\terminals, \nonterms,
\rules, \startsym)$ where each terminal $t \in T$ is of the form
``\code{api\_call $\$i_{1}$.m($\$i_{2}$, ..., $\$i_{p}$)}'', $n \in N$ is a non-terminal, $R$ is a set of productions,
and $S$ is the start symbol. Given 
grammar $\cfgspec$, we write $T_m$ to denote the subset of terminals 
involving a call to method $m$. As
mentioned in Section~\ref{sec:motiv}, each $\$i_{j}$ is a so-called
\emph{wildcard} that  can match any value of the appropriate type. To omit explicit type declarations, we assume
the existence of a typing oracle
$\typeor$ that
returns the type of a wildcard $w$, and, as standard, we use the notation
$\typeor \vdash w : \uptau$ to indicate that $w$ is of
type~$\uptau$.  
We also define a function to 
extract all wildcard symbols that appear in the grammar:

\begin{definition} {\bf (Wildcard extractor, $\wildcardsof$)} 
Given a
context-free protocol  $\cfgspec =  (\terminals, \nonterms,
\rules, \startsym)$, we write $\wildcardsof(\cfgspec)$ to
denote the set of all wildcard symbols that appear in $\cfgspec$. 
\end{definition}

\subsection{Semantic Conformance to API Protocol}\label{sec:conform}

Intuitively, a program $\program$ conforms to a parametrized CFG specification $\cfgspec$ if it satisfies the spec for every possible instantiation of the wildcards in $\cfgspec$. To make this statement more precise, we first introduce the notion of an \emph{instantiated API protocol}:

\begin{definition} {\bf (Instantiated spec)} Given an API  specification
  $\cfgspec$, we say that $\instspec$ is an instantiation of $\cfgspec$, written  $\instspec \in \inst(\cfgspec)$,  if it can be obtained from $\cfgspec$ by substituting every wildcard symbol $w_i \in \wildcardsof(\cfgspec)$  with a concrete value of the appropriate type. 
  \end{definition}

  Next, to determine if a program trace $\tau$ conforms to an instantiated specification $\instspec$, we will check ``inclusion" of the trace in the language defined by $\instspec$. To this end, we convert the trace to a word over the terminal symbols in $\instspec$ using the following \emph{TraceToWord} function:

\begin{definition}{\bf (Trace-to-Word)} Let $\ptrace$ be a trace and let $\instspec =  (\terminals, \nonterms,
\rules, \startsym)$  be an (instantiated) API protocol. We define $\project(\ptrace, { \instspec})$ as follows\footnote{We use the notation $[ s \mid ... ]$ to describe a filter operation on the input trace. The output preserves the relative order of statements in the input trace.}:
\[
  \project(\ptrace, \instspec) = [ s' \mid s' \in \terminals, \langle s,\ \sigma \rangle \in \ptrace,\ s' = s[\sigma(\vec{v})/\vec{v}],\ \vec{v} = \textsf{Vars}(s)]
\]
\end{definition}

\begin{example}

    Consider the following  trace $\ptrace$:
    \begin{equation*}
      \begin{split}
        \ptrace = \langle \code{l1 = new Lock}, \pstate_1 \rangle, \langle \code{l1.lock()}, \pstate_2 \rangle, \langle \code{l1.unlock()}, \pstate_3 \rangle,\\
        \langle \code{l2 = new Lock}, \pstate_4 \rangle, \langle \code{l2.lock()}, \pstate_5 \rangle, \langle \code{l2.unlock()}, \pstate_6 \rangle
      \end{split}
    \end{equation*}
and suppose that $o_1, o_2$ refer to the addresses of the first and second allocated \code{Lock} objects respectively.
Now, consider the following instantiated spec $\instspec$:
\[ 
  \instspec \ = \    S \ \rightarrow \ \epsilon \ | \ o_1.lock() \ S \ o_1.unlock() \ S
\]
Then, we have:
\[
\project(\ptrace, \instspec) = [\code{$o_1$.lock()}, \code{$o_1$.unlock()}] \\
\] 

Observe that the generated word ``ignores'' all statements other than API calls (e.g., \code{new Lock}). Furthermore, since 
variable \code{l2} has value $o_2$ rather than $o_1$, the last two \code{lock}/\code{unlock} statements in the trace are also not included in the result. \\

\end{example}

\definition{\bf (Semantic conformance)} Given a program $\program$ and a
context-free API protocol $\cfgspec$, $\program$ semantically conforms to 
$\cfgspec$ if and only if the following holds:
\begin{equation}\label{eq:probstmt}
  \forall \ptrace \in \emph{Traces}(\program).\forall \instspec \in \inst(\cfgspec).\ \project(\ptrace, \instspec) \in \lang(\instspec)
\end{equation}

In other words, a program $\program$ satisfies $\cfgspec$ if it satisfies the protocol for all possible instantiations of the wildcards in $\cfgspec$ for every program trace.

\section{Program Instrumentation}\label{sec:instr}

\begin{figure*}
  \[
  \begin{array}{cc}
    (API) &
    \irule{
      \begin{array}{c}
        \terminals_m = \{t_1,...,t_k\}\ \ g_i = \getguard(t_i, s)\\
        s' = \code{if (} g_1 \code{) } t_1 \code{ ... else if(} g_k \code{) } t_k
      \end{array}
    }{
     \typeor, \cfgspec \vdash  s = \code{api\_call } v.m(\vec{v}) \hookrightarrow s'
    }\\\\
        (Seq) &
    \irule{
      \begin{array}{c}
     \typeor, \cfgspec \vdash  s_1 \hookrightarrow s_1' \ \ \ \ 
          \typeor, \cfgspec \vdash   s_2 \hookrightarrow s_2'
      \end{array}
    }{
     \typeor, \cfgspec \vdash  s_1; s_2 \hookrightarrow s_1'; s_2'
    }\\\\
            (\emph{If}) &
    \irule{
      \begin{array}{c}
     \typeor, \cfgspec \vdash  s_1 \hookrightarrow s_1' \ \ \ \ 
          \typeor, \cfgspec \vdash    s_2 \hookrightarrow s_2' 
      \end{array}
    }{
     \typeor, \cfgspec \vdash \code{if}(p) \ \{s_1\} \ \code{else} \{ s_2\}  \hookrightarrow \code{if}(p) \ \{s_1'\} \ \code{else} \  \{ s_2'\} }\\\\
    (Method) &
    \irule{
      \typeor, \cfgspec \vdash s \hookrightarrow s'
    }{
    \typeor, \cfgspec \vdash   \code{void } m(\vec{v})\code{\{} s\code{\}} \hookrightarrow \code{void } m(\vec{v})\code{\{} s' \code{\}}
    }\\\\
    (Class) &
    \irule{
      \begin{array}{c}
         
         w_i \in \wildcardsof(\cfgspec) \ \ \ \ \   \typeor \vdash w_i : \uptau_i \ \ \ \ \ \     f'_i = \code{static } w_i : \uptau_i = \code{*} \ \ \ \ \ \ 
         \typeor, \cfgspec \vdash m_i \hookrightarrow m_i'\ \\
      \end{array}
    }{
     \typeor, \cfgspec  \vdash  \code{class C \{ } f_1\ ...\ f_n\ \ m_1\ ...\ m_k \code{ \}} \hookrightarrow \code{class C \{ } f_1\ ...\ f_n\ f'_1\ ...\ f'_j\ \ m_1'\ ...\ m_k' \code{ \}}
    }
  \end{array}
  \]
  \caption{Rules for instrumenting program $\program$ for a given specification $\cfgspec = (\terminals, \nonterms, \rules, \startsym)$. For statements that are not shown, we have $\typeor, \cfgspec \vdash s \hookrightarrow s$, and the definition of $\getguard$ function is inlined in text.}\label{fig:instr}
\end{figure*}

In the previous section, we defined conformance of a program to an API protocol in terms of all possible program traces and all possible instantiations of the wildcard symbols. While this strategy allows us to formally state the problem, it does not lend itself to a verification algorithm since there are infinitely many possible instantiations of the wildcard symbols as well as infinitely many  program traces. Thus, rather than checking the containment of each trace in all possible instantiations of the parametrized CFG, our strategy is  to instead generate a CFG encoding {all possible}  traces of the program as well as all possible instantiations of the wildcard symbols and then check inclusion between this CFG and the specification grammar. Towards this goal, we first \emph{instrument} the program with new fields that are initialized non-deterministically and that can be used to capture all possible values of the wildcards in the specification.   In addition, our instrumentation deals with challenges that arise from potential aliasing
between different arguments to API calls.

 In more detail, Figure~\ref{fig:instr} describes our program instrumentation using judgments of 
the form $\typeor, \cfgspec \vdash s \hookrightarrow s' $, where $s'$ corresponds to the transformed version of $s$.

\paragraph{Class.} The top-level rule labeled  ``$Class$''  
introduces a static field for every wildcard
symbol that appears in   $\cfgspec$ and
initializes it to a non-deterministic value.  It also
instruments each  method within this
class.

\paragraph{Method, Seq, If.} These three rules  reconstruct the
statement after recursively transforming the statements nested inside them.

\paragraph{API} This rule is the core of our program instrumentation and 
ensures that each terminal symbol in the specification grammar has a (syntactically) corresponding
API call statement while being semantically equivalent to the original API call.
As shown in Figure~\ref{fig:instr}, this rule transforms an API call $s$ to library method $m$
to an \code{if}-\code{then}-\code{else} statement.  Specifically, the rule
iterates over all the terminals $t_k \in \terminals_m$ in $\cfgspec$ and
generates an \code{if} statement for each terminal $t_k$ conditioned
upon the wildcard symbols matching the variables used in $s$.  To achieve this goal, we make use of an
auxiliary $\getguard$ function defined as follows: 
\[ \getguard(t_k, s) =
\bigwedge_{j} \$i_{kj} = v_j\]  
Here, $\vec{\$i_k}$ is the sequence of
wildcards used in $t_k$ and $\vec{v}$ is the sequence of variables used
in $s$. Thus, given an API call $s$ and a set of terminals
$\terminals_m$, we generate the following code:
\begin{equation*}
\begin{split}
  & \code{if}(\$i_{11} = v_1 \ \land \  \ldots \ \land \  \$i_{1n} = v_n)\ \{ \  t_1 \  \} \ \\
  & \ldots \\
  & \code{else if}(\$i_{k1} = v_1 \ \land \  \ldots \ \land \  \$i_{kn} = v_n)\ \{ \  t_k \  \}\\
\end{split}
\end{equation*} 
Hence, our instrumentation ensures  that  API calls syntactically use the
wildcard symbols in the grammar while preserving program behavior relevant to the specification.

The following theorem states the correctness of our instrumentation:

\begin{theorem}\label{thm:equi-safe}
Let $\program$ be a program and $\cfgspec$ a context-free API protocol.
  If we have \ $\typeor, \cfgspec \vdash \program \hookrightarrow
  \program'$ and $\ \program'$ semantically conforms $\cfgspec$, then so does
  $\ \program$.
\end{theorem}

\begin{proof}
  \iffull{
    The proofs of all theorems are in the appendix.
  }
  \fi
\end{proof}

Observe that the above theorem only states the soundness, but not
completeness, of our program instrumentation. Completeness does not
hold for arbitrary parametrized CFGs. For example, consider the API
protocol: $\cfgspec \rightarrow \$1.f()\ \$2.g()$, where $\$1$ and
$\$2$ have different types, and the code fragment ``\code{v1.f()}
\code{v2.g()}''. This fragment clearly conforms to the API protocol,
however, our instrumentation would produce the following output:

\begin{center}
  \vspace{0.5em}
\begin{tabular}{l}
  \code{\$1 = *; \$2 = *;}\\
  \code{if (v1 == \$1) \$1.f();}\\
  \code{if (v2 == \$2) \$2.g();}
\end{tabular}
  \vspace{0.5em}
\end{center}

The instrumented program does not satisfy the API protocol because it
generates the words ``\$1.f()'' and ``\$2.g()'' that do not belong in
$\lang(\cfgspec)$. Such protocols typically do not occur in practice
because such examples refer to relationships between methods
defined in different classes, so this is no longer a protocol for a
single API.

However, completeness \emph{does} hold if all terminals in the grammar
use the same set of wildcards. In practice, every API protocol we have
encountered conforms to this restriction.

\section{Verification Algorithm}\label{sec:verification}

Our verification algorithm takes as input a program that has been instrumented as described in Section~\ref{sec:instr}. The main idea underlying the algorithm is to extract a context-free grammar from the instrumented program and iteratively refine this CFG abstraction until the property is either refuted or verified. Since our algorithm operates over 
 \emph{predicated control flow automata (PCFA)}, we start with a discussion of PCFAs and then describe our CEGAR-based verification approach.

\subsection{Predicated Control-Flow Automata}\label{sec:pcfa}

We represent each program using a generalized form of \emph{control flow automaton (CFA)} that is commonly used in software model checking~\cite{love-automata,lazy-abs,abs-from-proofs}. A CFA is  a directed graph where nodes correspond to program locations, and an edge from $n$ to $n'$ labeled with $s$ indicates that the program transitions from location $n$ to $n'$ upon the execution of statement $s$. 
Predicated control flow automata (PCFA)  augment CFA nodes with logical predicates:

\definition{\bf (PCFA)} 
A predicated control-flow
automaton $\cfa$ is a tuple $\cfa = (\stmts, \states,
\transrel)$ where:
\begin{itemize}[leftmargin=*]
  \item $\stmts$ is the set of atomic program statements. 
  \item $\states$ is a set of  states, where each
    $\s \in \states$ is a pair $\s = (\ploc_{\meth}, \pred)$. Here,
    $\ploc_{\meth}$ is a program location within method $\meth$, and
    $\pred$ is a formula over some first-order theory.
  \item $\transrel$ is the transition relation  $\transrel
    \subseteq \states \times \stmts \times \states$.
\end{itemize}

\paragraph{Notation.} Given a state $s = (\ploc, \pred)$, we use $\locof(s)$ and
$\predof(s)$ to denote $\ploc$ and $\pred$
respectively. $\transof(\cfa)$ denotes the transition relation of
$\cfa$. We  use the notation $\satloc{S}{\ploc} = \{s \in S \mid
\locof(s) = \ploc\}$ to represent the subset of states in $S$ that
involve program location $\ploc$.  In addition, we write  $\inof(\ploc, \transrel)$ 
(resp. $\outof(\ploc, \transrel)$) to denote the in-coming (resp. out-going) edges of 
location of $\ploc$. Finally, we say that state $s'$ is reachable from state $s$, denoted as
$\cfa \vdash s \leadsto s'$, if and only if $(s, \_, s') \in
\transrel$. As standard, we use $\cfa \vdash s \leadsto^{*} s'$ to
represent the  transitive closure of relation $\leadsto$.

\subsection{Main Algorithm}\label{sec:main-alg}

Figure~\ref{fig:main-alg} presents  our top-level 
verification algorithm. This procedure 
 takes as input an (instrumented)
program $\program$, represented as  a mapping from methods to their PCFAs, as well as a
context-free API protocol $\cfgspec$. The algorithm  either returns ``Verified'' or a
counterexample  indicating an  API misuse. As a convention, procedure names in small caps are formally defined later in this
paper, whereas those in camel case are oracles that provide
 functionality that is orthogonal to our approach. 
 
\begin{figure}
\begin{algorithm}[H]
  \begin{algorithmic}[1]
    \Procedure{Verify}{$\program$, $\cfgspec$}
      \State \textbf{input:} $\program: \meths \rightarrow \cfas$, program.
      \State \textbf{input:} $\cfgspec$, API-Protocol's context-free grammar.
      \State \textbf{output:} $\mathit{Verified}$ or Counterexample.
      \vspace{0.04in}

      \While{$true$} 
        \State $\cfgrammar_{\program} \gets \textsc{ConstructCFG}(\program)$ \label{ln:cfg-const}
        \If{$\exists d.\ d \in InclusionCheck(\cfgprog, \cfgspec)$} \label{ln:incl-check}
          \State $(\pi, \rightsquigarrow) \gets derivation2path(d)$ \label{ln:path-const}
            \If{$\emph{feasible}( (\pi, \rightsquigarrow)) $} \label{ln:feas-check}
             \Return $\pi$ 
          \Else 
            \State $\mathcal{I} \gets \mathit{Interpolant}((\pi, \rightsquigarrow))$ \label{ln:interp} 
            \State $\Psi \gets \left \{\ploc_{\meth} \mapsto \left \{I_{j} \bigm \vert \begin{array}{c}I_{j} \in \mathcal{I}, \sigma_j \in \pi\\ \locof(\sigma_j) = \ploc_{\meth}\end{array} \right \} \right \}$ \label{ln:inter-map}
            \State $\program \gets \textsc{Refine}(\program, \Psi)$ \label{ln:refine}
          \EndIf
        \Else \ \Return $\mathit{Verified}$ \label{ln:verified}
        \EndIf
      \EndWhile
    \EndProcedure
  \end{algorithmic}
\end{algorithm}
\vspace{-2.2em}
\caption{Verification Algorithm}\label{fig:main-alg}
\end{figure}

The main verification algorithm  is a CEGAR loop that consists of the
following steps. First, it calls procedure {\sc{ConstructCFG}}
(line~\ref{ln:cfg-const}) to obtain a context-free grammar $\cfgprog$ that
abstracts the relevant API usage  of $\program$. Next, it checks whether there
exists a word $w$ that belongs in $\lang(\cfgprog)$ but not in
$\lang(\cfgspec)$ (line~\ref{ln:incl-check}). If this is not the case,
 the  program must satisfy  $\cfgspec$, so the
algorithm returns ``Verified''
(line~\ref{ln:verified}).

On the other hand, if there exists a word $w \in \lang(\cfgprog) \backslash \lang(\cfgspec)$,
we need to check whether $w$ corresponds to a feasible execution path of $\program$. Given a derivation $d$ of 
$w$, we  convert this derivation  to an execution path using an oracle called 
\emph{derivation2path} (line \ref{ln:path-const}). Here, we represent an execution path  as a \emph{nested trace}~\cite{nested-interp}, which
 is  a tuple
  $(\pi = \sigma_0 ... \sigma_n,
\rightsquigarrow)$ where $\pi$ is a sequence of program statements
and $\rightsquigarrow$ is a so-called "nesting relation" between indices of $\pi$ that
associates matching call and return statements. That is, if $i
\rightsquigarrow j$, then $\sigma_j$ is a return statement and
$\sigma_i$ is its matching call statement. Given such a nested 
  trace, 
we can easily check whether $\pi$ is feasible by encoding it as an SMT formula
and querying its satisfiability
(line~\ref{ln:feas-check}).  If
the path is  feasible, then the algorithm returns $\pi$ as a witness
of API misuse.

In case $\pi$ is infeasible, then word $w$ is a spurious
counterexample, and our algorithm refines the PCFA abstraction
(lines~\ref{ln:interp}-\ref{ln:refine}) to eliminate the same spurious
counterexample in the next iteration. To this end, we first make use
of another oracle, $\mathit{Interpolant}$, which takes as input a
nested word $(\pi = \sigma_0 ... \sigma_n, \rightsquigarrow)$ and
returns an \emph{inductive sequence of nested interpolants}
$\mathcal{I} = [I_0, ..., I_{n+1}]$. Following~\citet{nested-interp},
we define nested interpolants as a sequence of predicates with the
following properties: (1)~$I_0 = true$, $I_{n+1} = \emph{false}$. (2)~If
$\sigma_i$ is not a return statement, then $sp(\sigma_i, I_i)
\Rightarrow I_{i+1}$. (3) If $\sigma_i$ is a return statement, then
$sp(\sigma_i, I_i \land I_j) \Rightarrow I_{i+1}$ and $j
\rightsquigarrow i$. Intuitively, the first property ensures that $I$
can be used to prove infeasibility of $(\pi, \rightsquigarrow)$,
whereas the latter two properties ensure that $I$ is inductive.

After calculating a nested interpolant, the algorithm builds a mapping $\Psi$ that groups interpolants by program location
(line~\ref{ln:inter-map}). That is, $\Psi$ maps each program location to a set of predicates that should be tracked at that location. The {\sc Refine} procedure uses  $\Psi$ to determine how to clone program locations in the PCFAs such that $(\pi, \rightsquigarrow)$ is  no longer  feasible in the refined program abstraction.

We now state the following two theorems concerning the soundness and progress of our approach:

\begin{theorem}\label{thm:soundness}
  {\bf{(Soundness)}} 
  Let $\program, \program'$ be the  programs before and after
  the call to {\sc Refine} at line 13 respectively. Then, for every feasible execution path $\pi$ in
  $\program$, there exists a derivation $d \in
  \textsc{ConstructCFG}(\program')$ such that $(\pi, \rightsquigarrow)
  = \mathit{derivation2path}(d)$.
\end{theorem}

\begin{proof}
  \iffull{
    The proofs of all theorems are in the appendix.
  }
  \fi
\end{proof}

\begin{theorem}\label{thm:completeness}
  {\bf{(Progress)}} Let $t$ be  a spurious counterexample
  returned by $\mathit{derivation2path}$ and let $\program'$ be the resulting
  program after calling {\sc Refine} on program $\program$.
  Then, there does not exist a derivation $d \in
  \textsc{ConstructCFG}(\program')$ such that $t =
  \mathit{derivation2path}(d)$.
\end{theorem}
\begin{proof}
  \iffull{
    The proofs of all theorems are in the appendix.
  }
  \fi
\end{proof}

In the following subsections, we describe the {\sc{Refine}}
(Section~\ref{sec:pcfa-ref}) and {\sc{ConstructCFG}}
(Section~\ref{sec:cfg-const}) procedures in more detail.

\subsection{PCFA Refinement}\label{sec:pcfa-ref}

Our PCFA refinement algorithm is summarized in 
Figure~\ref{fig:ref-alg}.  Given program $\program$ and  mapping $\Psi$ from  locations to predicates,
the  idea is to "clone" any program location $\ploc \in \textsf{dom}(\Psi)$ 
based on the predicates $\Psi(\ploc)$.
Intuitively,
the demand-driven cloning of program locations allows our method to be selectively path-sensitive and removes 
infeasible program paths encountered in previous iterations. Furthermore, our refinement algorithm is modular in the sense that
we can refine the PCFA of each method independently.

\begin{figure}
\begin{algorithm}[H]
  \begin{algorithmic}[1]
    \Procedure{Refine}{$\program$, $\Psi$}
      \State \textbf{input:} $\program: \meths \rightarrow \cfas$, program.
      \State \textbf{input:} $\Psi: \plocs \rightarrow \{\preds\}$, new predicates to track.
      \State \textbf{output:} Refined program with respect to $\Psi$
      \vspace{0.04in}

      \For{$(\ploc_{\meth}, \mathit{Preds}) \in \Psi$}
        \State $(\stmts, \states, \transrel) \gets \program[\meth]$ \label{ln:get-pcfa}
        \State $\Phi \gets \ccubes(\mathit{Preds})$ \label{ln:cubes}
        \State $\states' \gets \textsf{CloneStates}(\states, \ploc_{\meth}, \Phi)$ \label{ln:clone}
        \State $\transrel' \gets \textsf{UpdateTransitions}(\transrel, \ploc_{\meth}, \satloc{\states'}{\ploc_{\meth}})$ \label{ln:trans}
        \State $\program[\meth] \gets (\stmts, \states', \transrel')$ \label{ln:update-pcfa}
      \EndFor
      \State \Return $\program$
    \EndProcedure
  \end{algorithmic}
\end{algorithm}
\vspace{-2.2em}
\caption{Program Refinement Algorithm.}\label{fig:ref-alg}
\end{figure}

In more detail, the {\sc Refine} procedure iterates over each program location $\ploc \in \textsf{dom}(\Psi)$ and determines which new states to create in the PCFA. Specifically, if $\Psi(\ploc)$ contains $n$ new predicates, then, for each state $(\ploc, \phi)$ in the PCFA, we need to create $2^n$ new states, where each clone represents a copy of $\ploc$ under a different boolean assignment to the predicates in $\Psi(\ploc)$. Towards this goal, the {\sc Refine} procedure first invokes  \ccubes \ (line~\ref{ln:cubes}) to generate a different boolean assignment as follows:
\begin{equation*}
\ccubes(P) = \big \{\underset{i = 1}{\overset{|P|}{\bigwedge}} c_{i} \mid c_{i} \in \{p_i, \neg p_i\}, p_i \in P \big \}
\end{equation*}
In other words, \ccubes($P$)\ yields a set $\Phi$ of (conjunctive) formulas such that every $\phi \in \Phi$ corresponds to a different boolean assignment to the predicates in $P$.

Next, given the new set of predicates $\Phi$ 
to track at location $\ploc$, the procedure {\sc CloneStates} (line~\ref{ln:clone}) generates $|\Phi|$ clones of each state $(\ploc, \phi) \in S$ as follows: 
\begin{equation*}
\textsf{CloneStates}(\states, \ploc, \Phi) = (\states \setminus \satloc{\states}{\ploc}) \  \cup \ \{(\ploc, \pred \land \pred') \mid (\ploc, \pred) \in \states,\ \pred' \in \Phi \}
\end{equation*}

In other words, {\sf CloneStates} removes all existing states $(\ploc, \phi)$ associated with location $\ploc$ and then adds a new state $(l, \phi \land \phi')$ for each $\phi' \in \Phi$. Thus, if the PCFA contains $n$ states for location $\ploc$ before refinement, then the refined PCFA contains $n \times |\Phi|$ states for location $\ploc$.

\begin{example}
Consider the initial PCFA
for method \code{acquire}  from Fig.~\ref{fig:init-acquire} and 
suppose  $\Psi(a_3) = P = \{ l1_{acq} = \$1 \}$. In this case, 
we have $\Phi = \ccubes(P) = \{ l1_{acq} = \$1, l1_{acq} \neq \$1\}$. Thus, {\sf CloneStates}
removes the original state $(a_3, \emph{true})$ and generates two new states
$(a_3, l1_{acq} \neq \$1)$ and $(a_3', l1_{acq} = \$1)$ as shown in Figure~\ref{fig:pcfaref}.
\end{example}

After creating the new states $S'$, the {\sc Refine} procedure updates the
transition relation of the PCFA by invoking the {\sf{UpdateTransitions}} function (line~\ref{ln:trans}), defined
as follows:

\begin{equation*}
  \begin{split}
    \textsf{UpdateTransitions}(\transrel, \ploc, \states') = & \ \transrel \setminus (\inof(\transrel, \ploc) \cup \outof(\transrel, \ploc))\ \cup\\
     & \{e = (s, \sigma, s') \mid s' \in S', (s, \sigma, \_) \in \inof(\transrel, \ploc),\ \mathit{feasible}(e)\} \ \cup\\
     & \{e = (s', \sigma, s) \mid s' \in S', (\_, \sigma, s) \in \outof(\transrel, \ploc),\ \mathit{feasible}(e)\}
  \end{split}
\end{equation*}
where $\mathit{feasible}((s_1, \sigma, s_2))$ is defined as $SAT(sp(\sigma, \predof(s_1)) \land \predof(s_2))$.
In other words, \textsf{UpdateTransitions} first removes from $\transrel$ all transitions involving location $\ploc$. Then, for each new state $s' \in S'$ and for each incoming  edge $(s, \sigma, \_)$  to location $\ploc$, it adds a new edge $(s, \sigma, s')$ as long as the annotation  of the new state $s'$ is consistent with the annotation of the source node, $\emph{Pred}(s)$, and the semantics of statement $\sigma$. Outgoing edges from location $\ploc$ are also updated analogously.

\begin{example}
Consider again the new states at the end of method
\code{acquire}. Observe that {\sf{UpdateTransitions}} will not add an
edge between states $a_1$,$a_3'$ and $a_2$,$a_3$ in the refined
version of the PCFA (shown in Fig.~\ref{fig:acq-ref}) because
$\mathit{feasible}$  returns false for these edges.
\end{example}

\subsection{Context-Free Grammar Construction}\label{sec:cfg-const}

\begin{figure}
\begin{algorithm}[H]
  \begin{algorithmic}[1]
    \Procedure{ConstructCFG}{$\program$} \State \textbf{input:}
    $\program: \meths \rightarrow \cfas$, program.  \State
    \textbf{output:} $\cfgprog$, context-free grammar that abstracts
    $\program$.
      \vspace{0.04in}

      \State $(\terminals, \nonterms, \rules) \gets (\emptyset, \emptyset, \emptyset)$
      \State $\Theta \gets \{ (s, m) \mid s \in \finsof(\program[\meth])\}$ \label{ln:clones}
      \vspace{0.05in}
      \For{$(s_i, m) \in \Theta$}
        \State $( \terminals_i, \nonterms_i, \rules_i, \startsym_i) \gets \textsc{GenGrammar}(\program[\meth], s_i, \Theta)$ \label{ln:clone-cfg}
        \State $\terminals \gets \terminals \cup \terminals_i,\ \nonterms \gets \nonterms \cup \nonterms_i,\ \rules \gets \rules \cup \rules_i$ \label{ln:union}
        \If{\textsf{IsMain}($m$)}  $\startsym \gets \startsym_i$   
        \EndIf
      \EndFor
      \State \Return $(\terminals, \nonterms, \rules, \startsym)$
     \EndProcedure
     \vspace{0.04in}
  \end{algorithmic}
\end{algorithm}
\vspace{-2.2em}
\caption{Context-Free Grammar Construction}\label{fig:cfg-const}
\end{figure}

In this section, we describe how to extract a context-free grammar from the PCFAs.  As explained earlier, the main idea is to represent  relevant API invocations as
terminals in the grammar  so that words generated by the CFG correspond to all possible sequences of API calls issued by the program. Towards this goal,  we introduce one non-terminal symbol for each PCFA state  and generate  CFG productions according to the PCFA transitions. The resulting CFG abstraction is (selectively) path-sensitive in that we introduce as many non-terminal symbols for a method as it has exit states. Intuitively, different  non-terminals for  method $m$ correspond to different "summaries" conditioned upon facts that hold at $m$'s call sites.

The {\sc ConstructCFG} procedure is described in more detail  in Figure~\ref{fig:cfg-const}. 
It generates the program's CFG abstraction by  iterating over every exit state $s$ of each method $m$ and constructs a separate grammar for $(s,m)$ using the call to {\sc GenGrammar} at line~\ref{ln:clone-cfg}. The CFG  for the whole program is obtained as the union of all of the individual grammars, and the start symbol for $\cfgprog$ is the one associated with \code{main}.

\begin{figure}
  \[
  \begin{array}{lc}
    (1) &
    \irule{
      \begin{array}{c}
        (s, \sigma, s') \in \transof(\cfa) \ \ \ \ \ \neg \mathit{callStmt}(\sigma)\ \ \ \ \
        \cfa  \vdash s' \leadsto^{*} c \ \ \ \  \  \emph{Pred}(c) = \varphi
       \end{array}
    }{
      \cfa,  c , \Theta \vdash  \{\mathcal{S}_{\pred}, \mathcal{S}_{\pred}'\} \subseteq \nonterms_{c}\ \ \ \ \mathcal{S}_{\pred} \rightarrow \mathcal{S}_{\pred}' \in \rules_{c}
    }\\\\
    (2) &
    \irule{
      \begin{array}{c}
        (s, \sigma, s') \in \transof(\cfa) \ \ \ \ \sigma = api\_call\ \meth(\vec{v}) \ \ \ \
        \cfa \vdash s' \leadsto^{*} c \ \ \  \  \emph{Pred}(c) = \varphi
      \end{array}
    }{
      \begin{array}{c}
        \cfa, c, \Theta \vdash \{\mathcal{S}_{\pred}, \mathcal{S}_{\pred}'\} \subseteq \nonterms_{c}\ \ \sigma \in \terminals_{c}\ \ \mathcal{S}_{\pred} \rightarrow \sigma\ \mathcal{S}_{\pred}' \in \rules_{c}
      \end{array}
    }\\\\
    (3) &
    \irule{
      \begin{array}{c}
        \emph{feasible}(e, \pred') \ \ \  \  \emph{Pred}(c) = \varphi\\
        e = (s, \sigma, s') \in  \transof(\cfa)\ \ \ \ \sigma = call\ \meth'(\vec{v})\ \ \ \
        (c', m') \in \Theta \ \ \ \ \cfa \vdash s' \leadsto^{*} c\ \ \ \ \pred' = \predof(c')
      \end{array}
    }{
      \begin{array}{c}
         \cfa, c, \Theta \vdash  \{\mathcal{S}_{\pred}, \mathcal{S}_{\pred}'\} \subseteq \nonterms_{c}\ \ \ \ \mathcal{S}_{\pred} \rightarrow \methrule_{\pred'}'\ \mathcal{S}_{\pred}' \in \rules_{c}
      \end{array}
   }
   \\\\
   (4) &
   \begin{array}{ccc}
    \irule{
        s \in \initsof(\cfa)\ \ \ \ \cfa \vdash s \leadsto^{*} c \ \ \ \  \varphi = \emph{Pred}(c) 
    }{
     \cfa, c, \Theta \vdash  \methrule_{\pred} \rightarrow \mathcal{S}_{\pred} \in \rules_{c}\ \ \ \ \methrule_{\pred} \in \nonterms_{c} \ \ \ \ \ S_c =   \methrule_{\pred} 
    }  &
    (5) &
    \irule{\varphi = \emph{Pred}(c)}{
      \begin{array}{c}
        \cfa, c, \Theta \vdash  \mathcal{C}_{\pred} \rightarrow \epsilon \in \rules_{c}
      \end{array}
    }
  \end{array}
  \end{array}
  \]
  \caption{Rules for constructing CFG = $(\terminals_c, \nonterms_c, \rules_c, \startsym_c)$ given a exit state $c$ in PCFA $\cfa = (\stmts, \states, \transrel)$, and set $\Theta$. For a  PCFA state $s$ with predicate $\varphi$, the symbol $\mathcal{S}_\varphi$ denotes the corresponding non-terminal in the grammar.} \label{fig:clone-cfg}
\end{figure}

Figure~\ref{fig:clone-cfg} summarizes the {\sc GenGrammar} procedure
using inference rules of the following shape: 
\[ \cfa, c, \Theta \vdash \Delta_1,
\ldots, \Delta_n \]
Here,   the left-hand side of the turnstile
represents the arguments of the {\sc GenGrammar} procedure, and each
$\Delta_i$ is a set inclusion constraint for the CFG symbols  and productions. In more
detail, $\cfa$ is the PCFA for the current method, $c$ is an exit
state in $\cfa$, and $\Theta$ is a set of pairs $(s, m)$ where $s$ is
an exit state in method $m$'s PCFA. (As we will see shortly, {\sc
  GenGrammar} uses $\Theta$ to generate grammar productions for method
calls.) Given a state $s$ in the PCFA and predicate $\varphi$ labeling
exit state $c$,  {\sc GenGrammar}  generates a non-terminal
$\mathcal{S}_\varphi$ for each state in the PCFA.

\paragraph{Statements} The  first rule in Figure~\ref{fig:clone-cfg} applies to all statements that are \emph{not} function calls.
Since atomic statements other than API calls are not relevant
to our abstraction, this rule only captures control-flow
dependencies.
Specifically, let $(s, \sigma, s')$ be a PCFA edge where $\sigma$ is a non-call
statement. First, we introduce non-terminals $\mathcal{S}_\varphi, \mathcal{S}'_\varphi$ for states $s, s'$ and  add a production $\mathcal{S}_\varphi \rightarrow \mathcal{S}'_\varphi$ to capture that $s'$ is a successor of $s$. Observe that this rule (as well as the next two rules) have $\cfa \vdash s' \leadsto^* c$ as a premise because  non-terminals $\mathcal{S}_\varphi, \mathcal{S}'_\varphi$ should only be added to the grammar if $s, s'$ are  backward-reachable from exit state $c$.

\begin{example}
The production $\mathcal{A}_{1,\phi_1} \rightarrow
\mathcal{A}_{2,\phi_1}$ in Fig.~\ref{fig:refcfg} is generated using
the \emph{Stmt} rule based on the PCFA from Fig~\ref{fig:pcfaref}.
\end{example}

\paragraph{API} The next rule generates productions for  API calls. This rule is  similar to the previous one but with two key differences: First,  it also adds $\sigma$ to terminals $T_c$. Second, it generates the production $\mathcal{S}_\varphi \rightarrow \sigma \mathcal{S}'_\varphi$ instead of $\mathcal{S}_\varphi \rightarrow \mathcal{S}'_\varphi$ because $\sigma$ is relevant to the program's API usage.

\begin{example}
Consider the production $\mathcal{A}_{2,\phi_1} \rightarrow$
\code{\$1.lock()} $\mathcal{A}_{3,\phi_1}'$ from
Figure~\ref{fig:refcfg}. This production is generated due to the PCFA
transition $(a_2,\code{\$1.lock()},a_3')$ from
Figure~\ref{fig:pcfaref}.
\end{example}

\paragraph{Call.} 
The third rule applies to PCFA edges $(s, \sigma, s')$ where $\sigma$ is a
call to  method $m'$. Since there
are multiple "clones" of $m'$, let us consider one specific clone $c'$
with "summary" $\varphi'$. In this case, we generate the production
$\mathcal{S}_\varphi \rightarrow \mathcal{M'}_{\varphi'}
\mathcal{S'}_\varphi$, where $\mathcal{M'}_{\varphi'}$ is the start
symbol for the grammar associated with this clone of $m'$. However,
since predicate $\varphi'$ may be inconsistent with PCFA transition
$(s, \sigma, s')$, we first check whether this \emph{particular} clone
of $m'$ is feasible at this  call site. This is
done by requiring \emph{feasible}$(e, \varphi')$, 
defined as follows:
\begin{equation*}
  \mathit{feasible}((s, \code{call m'($\vec{v}$)}, s'), \varphi') \equiv SAT(\predof(s) \land \predof(s') \land \varphi')
\end{equation*}

\begin{example}
  Consider the PCFAs from Figure~\ref{fig:pcfaref}. Here, the
  production $\mathcal{F}_3 \rightarrow
  \mathit{Acquire}_{\{l1_{acq}=\$1\}} \mathcal{F}_4$ belongs to
  $\cfgprog$ because we have \emph{feasible}($e, l1_{acq}=\$1$) for
  the PCFA edge $e$ from $f_3$ to $f_4$. On the other hand, there is
  no production $\mathcal{F}_3 \rightarrow
  \mathit{Acquire}_{\{l1_{acq}=\$1\}} \mathcal{F}_4'$ because
   $
l1_{acq} = l \land l \neq \$1 \land l1_{acq} = \$1
   $ is unsatisfiable.
\end{example}

\paragraph{Entry and exit.} The last two rules in Figure~\ref{fig:clone-cfg} deal with the entry and exit states of the PCFA. Specifically, for any entry state $s$ of the PCFA that is backward-reachable from the target exit state $c$, we add a production $\mathcal{M}_\varphi \rightarrow \mathcal{S}_\varphi$, where $\mathcal{M}_\varphi$ corresponds to the start symbol of the grammar. For exit state $c$, we just add the empty production $\mathcal{C}_\varphi \rightarrow \epsilon$.

\begin{example}
For the PCFA from Figure~\ref{fig:acq-ref}, we add the production
$\mathit{Acquire}_{\phi_1} \rightarrow \mathcal{A}_{0,\phi_1}$ because
$a_0$ is an entry state that is backward-reachable from state
$a_3'$. Similarly, we add a production $\mathcal{A}'_{3, \phi_1}
\rightarrow \epsilon$ for exit state $a_3'$.
\end{example}

\section{Implementation}\label{sec:impl}

We implemented our approach in a prototype  called \toolname\ for analyzing Java programs.
\toolname\ is
implemented in Java  on top of the Soot infrastructure~\cite{soot} and
uses the technique of~\citet{grammar-comp} to perform grammar inclusion
checks. Our implementation
also makes use of SMTInterpol~\cite{smt-interpol} to obtain nested
interpolants  and leverages  Z3~\cite{z3} to 
determine satisfiability.

In the remainder of this section, we discuss some design choices and
optimizations that were omitted from the technical presentation.

\paragraph{Slicing input programs.}
Before running the verification algorithm presented in
Section~\ref{sec:verification}, \toolname\ uses slicing to improve
scalability. Specifically, we first identify all calls to the API
whose usage is being checked and then compute a backward slice of the
program with respect to those
statements~\cite{weiser-slicing,thin-slicing}.

\paragraph{From words to execution paths.}
As mentioned in Section~\ref{sec:verification}, we assume that the
$\mathit{InclusionCheck}$ method returns a derivation $d \in \cfgprog$
for a word $w \in \lang(\cfgprog) \setminus \lang(\cfgspec)$. In
practice, $\cfgprog$ tends to be highly ambiguous, so obtaining such a
derivation for $w$ can be computationally expensive. To address this issue, we first convert $\cfgprog$ to Chomsky Normal Form
(CNF)~\cite{chomsky1959certain} for which 
there is a polynomial
algorithm for obtaining a derivation~\cite{hopcroft2008introduction}, and we then map this
derivation back to the original grammar. While  mapping  the CNF derivation to the 
original grammar is not polynomial time, we have found this strategy to work much better in practice compared to 
directly searching for a derivation in the original grammar.

\paragraph{Handling pointers.} 
In our implementation, we model the heap by using a fairly standard array-based 
encoding that has been popularized by ESC-Java~\cite{esc-java}. 
Specifically, we introduce an array for each field and model
loads and stores using select and update functions
in the theory of arrays.

\paragraph{Obtaining PCFAs.}
Before
generating the PCFA of a method, we first perform a program transformation
similar to the one described by~\citet{pol-pred-abs} to enable polymorphic predicate
abstraction. Specifically, for each method in the program, we 
generate auxiliary variables, referred to as \emph{symbolic constants} in 
prior work, that track the initial value of variables on method entry. 
This transformation allows computing polymorphic interpolants
that can be reused across  call sites.

\paragraph{Optimizations.}
Rather than introducing one non-terminal symbol for every
\emph{program location}, we instead introduce one non-terminal for
each \emph{basic block} in order to make the resulting context-free
grammar smaller. Also, since the refinement algorithm may issue an
exponential number of satisfiability queries, we issue SMT queries in
parallel whenever possible and memoize the results of Z3
queries. Finally, since mapping parse trees from the CNF 
grammar back to the original version can be a performance bottleneck,
 we memoize partial results between refinement
iterations.

\paragraph{Limitations.}
Similar to other verification tools, \toolname\ models several Java
features (e.g., exceptions, reflection) in a ``soundy''
way~\cite{soundiness}. Furthermore, since \toolname\ models program
semantics using the combined theory of arrays and linear integer
arithmetic, it also conservatively over-approximates operations that
fall outside of this theory. In particular, \toolname\ introduces
appropriate uninterpreted functions to model operations that involve
non-integer variables (e.g., floats, doubles, etc.).


\newcommand{\wifi}{WifiLock}
\newcommand{\wake}{WakeLock}
\newcommand{\loc}{LocationManager}
\newcommand{\rlock}{Simple ReentrantLock}
\newcommand{\crlock}{ReentrantLock}
\newcommand{\canvas}{Canvas}
\newcommand{\jsongen}{Json}

\newcommand{\ExoPlayerWifiBug}{ExoPlayer ($\lightning$)}
\newcommand{\ExoPlayerWifiFix}{ExoPlayer}
\newcommand{\ExoPlayerWakeBug}{ExoPlayer ($\lightning$)}
\newcommand{\ExoPlayerWakeFix}{ExoPlayer}
\newcommand{\ConnectBotBug}{ConnectBot ($\lightning$)}
\newcommand{\ConnectBotFix}{ConnectBot}
\newcommand{\Hystrix}{Hystrix}
\newcommand{\Guice}{Guice}
\newcommand{\Bitcoinj}{Bitcoinj}
\newcommand{\Glide}{Glide}
\newcommand{\Atlas}{Atlas}
\newcommand{\RxTool}{RxTool}
\newcommand{\Litho}{Litho}
\newcommand{\Hadoop}{Hadoop}
\newcommand{\HystrixJsonOne}{Hystrix-1}
\newcommand{\HystrixJsonTwo}{Hystrix-2}

\newcommand{\ExoPlayerWifiBugStats}{ 
  Cex  & 63.8        & 0.3              &  27 & 2/6  & 40}
\newcommand{\ExoPlayerWifiFixStats}{ 
  Safe & 35.2        & 0.4              &  20 & 1/4  & 44}
\newcommand{\ExoPlayerWakeBugStats}{ 
  Cex  & 66.6        & 0.3              &  24 & 2/6  & 40}
\newcommand{\ExoPlayerWakeFixStats}{ 
  Safe & 47.0        & 0.5              &  20 & 2/6  & 39}
\newcommand{\ConnectBotBugStats}{    
  Cex  & 392.0       & 9.3              &  42 & 3/9  & 107}
\newcommand{\ConnectBotFixStats}{    
  Safe & 2336.5      & 18.3             &  48 & 3/12 & 133}
\newcommand{\HystrixStats}{          
  Safe & 20.7        & 0.3              &  9  & 1/3  & 21}
\newcommand{\GuiceStats}{            
  Safe & 221.5       & 2.4              &  25 & 3/9  & 92}
\newcommand{\BitcoinjStats}{         
  Safe & 3175.3      & 28.0             &  79 & 1/5  & 115}
\newcommand{\GlideStats}{            
  Safe & 562.6       & 544.7            &  8  & 1/3  & 19}
\newcommand{\RxToolStats}{           
  Safe & 56.4        & 43.8             &  1  & 1/1  & 3}
\newcommand{\LithoStats}{            
  Safe & 14.4       & 0.5              &  5  & 1/3  & 11}
\newcommand{\HadoopStats}{           
  Safe & 140.2       & 64.5             &  49 & 1/4 & 65}
\newcommand{\HystrixJsonOneStats}{   
  Safe & 64.4        & 2.7              &  48 & 2/4 & 62}
\newcommand{\HystrixJsonTwoStats}{   
  Safe & 24.2        & 0.6              &  31 & 1/4 & 55}

\section{Evaluation}\label{sec:eval}

To evaluate \toolname, we collected  real-world use cases of
 Java APIs  and conducted experiments
designed to answer the following research questions:
\begin{itemize}
    \item[{\bf RQ1:}] Can \toolname\ verify  the correct  usage of popular  Java APIs in real-world clients?
    \item[{\bf RQ2:}]  Does the proposed technique advance the state-of-the-art in software verification?
\end{itemize}

To answer these questions, we conduct two sets of experiments. For our first experiment, we collect five popular 
Java APIs with context-free specifications and evaluate \toolname\
on 10 widely-used Java programs that leverage at least one of these five APIs. In our second experiment, we compare \toolname\ against existing verification tools. However, since there is no off-the-shelf  technique that can directly verify correct usage of context-free API protocols,
we instrument (simplified versions of) these 10 Java programs with suitable assertions that enforce correct API usage, and we then try to discharge these assertions using state-of-the-art verification and model checking tools.

All of our experiments are run on an Intel Xeon CPU
E5-2640 v3 @ 2.60GHz machine with 132 GB of memory running the Ubuntu
14.04.1 operating system.

\begin{table}
\small
  \centering
  \begin{tabular}{|c||l|}
  \hline
    {\bf API Name} & {\bf  Specification} \\
    \hline
    ReLock   &
\begin{tabular}{lll}
$S$ & $\rightarrow$ & $\$1.\code{acquire()}\ S\ \$1.\code{release()}\ S$ \\
    & $\ | \ $      & $\epsilon$
\end{tabular} \\
\hline
\begin{tabular}{l}
Wifi \& Wake Lock\\
\end{tabular}&
\begin{tabular}{lll}
$S$ & $\rightarrow$ & $RC \ | \ \$1.\code{setRefCnt(false)} \ NC$ \\\\

$NC$ & $\rightarrow$ & $\epsilon \ | \ NA \ \$1.\code{release()}$\\
$NA$ & $\rightarrow$ & $\$1.\code{acquire()} \ NA \ \$1.\code{release()} \ NA \ | \ \$1.\code{acquire()} \ NA $\\
            & $\ | \ $ & $\ \$1.\code{acquire()}$\\\\

$RC$ & $\rightarrow$ & $\$1.\code{acquire()} \ RC \ \$1.\code{release()} \ RC \ | \ \epsilon$\\
\end{tabular} \\
    \hline
Canvas & 
\begin{tabular}{lll}
$S$ & $\rightarrow$ & $\epsilon \ | \ \$1.\code{save()} \ S \ \$1.\code{restore()} \ S  $ \\
    & $\ | \ $ & $\$1.\code{save()} \ S$\\
\end{tabular} \\
    \hline 
\begin{tabular}{l} Json Gen.\end{tabular} & 
\begin{tabular}{lll}
$S$ & $\rightarrow$ & $\epsilon \ | \ Obj \ | \ Arr \ | \ \$1.\code{writeString()}$  \\
    & $\ | \ $ & $\$1.\code{writeNumber()} \ | \ \$1.\code{writeBoolean()}$\\\\

$Obj$ & $\rightarrow$ & $\$1.\code{writeStartObject()} \ Fld  \ \$1.\code{writeEndObject()}$\\\\
$Fld$ & $\rightarrow$ & $\$1.\code{writeFieldName()} \ S \ Fld \ | \ \epsilon$\\\\

$Arr$ & $\rightarrow$ & $\$1.\code{writeStartArray()} \ Vals \ \$1.\code{writeEndArray()}$\\\\
$Vals$ & $\rightarrow$ & $\epsilon \ | \ S \ Vals$
\end{tabular} \\
\hline
  \end{tabular} 
  \vspace{0.1in}
    \caption{Java API Protocol Specifications}  \label{tbl:java-specs}
\end{table}

\subsection{API Specifications \& Benchmarks}\label{subsec:specs-bench}

For our evaluation, we consider the following five popular  Java APIs  whose correct usage is defined by a context-free specification:

\begin{enumerate}
    \item {\bf ReentrantLock}: a widely-used Java API that implements a reentrant lock
    \item    {\bf WakeLock}: a popular Android API that  allows the client application to  keep the Android device awake
    \item {\bf WifiLock}: another Android API that  allows the applications to keep the Wi-Fi radio awake
    \item {\bf Canvas}:  a graphics API (also for Android) that allows clients to create views and animations
    \item {\bf JsonGenerator}: a serialization library that allows serializing Java objects as JSON documents
\end{enumerate}

\begin{table}
  \centering
  \footnotesize
  \begin{tabular}{l || l || 
      c c c c c c}
    \multicolumn{2}{c||}{Benchmark Info} & \multicolumn{6}{c}{\toolname\ Statistics}\\
    \hline
     & Benchmark & 
    Output & \makecell{Total \\ Time} & \makecell{Incl. \\ Check}  & \# Steps & \makecell{Preds / BB \\ (Avg/Max)} & \# Preds\\
    \hline
    &&&&&&&\\
    \multirow{2}{*}{\rotatebox[origin=c]{90}{Wifi}} & \ExoPlayerWifiBug & \ExoPlayerWifiBugStats \\
                             & \ExoPlayerWifiFix & \ExoPlayerWifiFixStats \\
    &&&&&&&\\
    \multirow{4}{*}{\rotatebox[origin=c]{90}{\wake}}   & \ExoPlayerWakeBug & \ExoPlayerWakeBugStats \\
                             & \ExoPlayerWakeFix & \ExoPlayerWakeFixStats \\
                             & \ConnectBotBug    & \ConnectBotBugStats \\
                             & \ConnectBotFix    & \ConnectBotFixStats \\
    &&&&&&&\\
    \multirow{3}{*}{\rotatebox[origin=c]{90}{\makecell{ReLock}}} & \Hystrix          & \HystrixStats \\
                             & \Guice            & \GuiceStats \\
                             & \Bitcoinj         & \BitcoinjStats \\
    &&&&&&&\\
    \multirow{4}{*}{\rotatebox[origin=c]{90}{\canvas}} & \Glide            & \GlideStats \\
                             & \RxTool           & \RxToolStats \\
                             & \Litho            & \LithoStats \\
    &&&&&&&\\
    \multirow{3}{*}{\rotatebox[origin=c]{90}{\jsongen}} & \Hadoop            & \HadoopStats \\
                             & \HystrixJsonOne            & \HystrixJsonOneStats \\
                             & \HystrixJsonTwo         & \HystrixJsonTwoStats \\
  \end{tabular}
  \vspace{1em}
  \caption{Results for \toolname. Under the ``output" column, "Cex"
    denotes a counterexample and "Safe" indicates that the benchmark
    was verified. Total time indicates end-to-end running time in
    seconds, and ``Incl. check'' shows the time spent performing
    grammar inclusion checking queries. ``\# Steps'': number of
    refinement steps, ``Preds / BB'': Average and max predicates
    tracked per basic block, ``\# Preds'': total number of predicates
    tracked.
  } \label{table:res-cfpchecker}
\end{table}

\begin{table}
  \centering
  \footnotesize
  \begin{tabular}{l || l || c c c || c c c || c c c}
        &      & \multicolumn{3}{c||}{JayHorn}        & \multicolumn{3}{c||}{JPF-BugFinder} & \multicolumn{2}{c}{JPF-Verifier}\\
    \hline
     & Bench. & Out. & \makecell{Correct \\ Output?} & Time & Out. & \makecell{Correct \\ Output?} & Time & Out. & \makecell{Correct \\ Output?} & Time\\
    \hline
    &&&&&&&&&\\
    \multirow{2}{*}{\rotatebox[origin=c]{90}{Wifi}} & \ExoPlayerWifiBug & Unknown & \xmark & 1212.4 & Safe & \xmark & 0.6 & TO & \xmark & - \\
                             & \ExoPlayerWifiFix & Unknown & \xmark & 1119.0 & - & - & - & TO & \xmark & - \\
    &&&&&&&&&&\\
    \multirow{4}{*}{\rotatebox[origin=c]{90}{\wake}}   & \ExoPlayerWakeBug & Unknown & \xmark & 1105.3 & Safe & \xmark & 0.5  & TO & \xmark & -\\
                             & \ExoPlayerWakeFix & Unknown & \xmark & 1162.8 & - & - & - & TO & \xmark & - \\
                             & \ConnectBotBug    & Unknown & \xmark & 333.0 & Safe & \xmark & 0.7 & TO & \xmark & - \\
                             & \ConnectBotFix    & Unknown & \xmark & 164.0 & - & - & - & TO & \xmark & - \\
    &&&&&&&&&&\\
    \multirow{3}{*}{\rotatebox[origin=c]{90}{\makecell{ReLock}}} & \Hystrix          & Unknown & \xmark & 12530.5 & & & & TO & \xmark & - \\
                             & \Guice            & Safe & \cmark & 2383.1 & - & - & - & TO & \xmark & - \\
                             & \Bitcoinj         & OM & \xmark & - & - & - & - & TO & \xmark & - \\
    &&&&&&&&&&\\
    \multirow{4}{*}{\rotatebox[origin=c]{90}{\canvas}} & \Glide            & TO & \xmark & - & - & - & - & TO & \xmark & - \\
                             & \RxTool           & TO & \xmark & - & - & - & - & TO & \xmark & - \\
                             & \Litho            & OM & \xmark & - & - & - & - & TO & \xmark & - \\
    &&&&&&&&&&\\
    \multirow{3}{*}{\rotatebox[origin=c]{90}{\jsongen}} & \Hadoop            & OM & \xmark & - & - & - & - & TO & \xmark & - \\
                             & \HystrixJsonOne          & Safe  & \cmark & 931.7 & - & - & - & TO & \xmark & - \\
                             & \HystrixJsonTwo         & OM & \xmark & - & - & - & - & TO & \xmark & - \\
  \end{tabular}
  \vspace{1em}
  \caption{Results for other safety-checking tools on simplified
    benchmarks using a time limit of 8 hours and memory limit of 16 GB per benchmark. Values in the ``Out.''  columns have the following
    meaning: Cex: feasible counterexample found, Safe: no violations
    found, Unknown: unable to produce neither a counterexample nor a
    proof of correctness, TO: timeout, OM: out of memory. We use a
    ``-'' to indicate that a value is not applicable. All execution
    times are in seconds.} \label{table:res-othertools}
\end{table}

\paragraph{Specifications.} Table~\ref{tbl:java-specs} presents the context-free protocols that
clients of these APIs must adhere to.  As used as a running example
throughout the paper, \code{ReentrantLock} requires calls
to \code{acquire} and \code{release} to be balanced, and failure to
follow this protocol results in deadlocks.  The next two APIs,
namely \code{WakeLock} and \code{WifiLock}, have the exact same
specification and can be used in two different modes of operation,
reference-counted and non-reference-counted. The specification for the
first mode is the same as \code{ReentrantLock} (i.e., each call
to \code{acquire} must be matched by a call to \code{release}). On the
other hand, the second mode is enabled by the
call \code{setRefCnt(false)} and requires the usage pattern to be of
the form $\code{acquire}^n \ \code{release}^m$ where $m \leq n$ and
$n \geq 1 \rightarrow m \geq 1$.  For both the \code{WakeLock}
and \code{WifiLock} APIs, failure to follow the protocol causes
resource leaks (e.g., the application drains the phone's battery).
For the Android \code{Canvas} API, its documentation states "It is an
error to call \code{restore()} more times than \code{save()} was
called."; thus, its specification is of the form
$\code{save}^n \ \code{restore}^m$ where $m \leq n$.  Failure to
follow this protocol results in a run-time exception.  The last API,
called \code{JsonGenerator}, has a relatively complex specification
and requires clients to call API methods
(e.g., \code{writeStartObject()}, \code{writeEndObject()}, etc.) in
accordance with the JSON schema, that is, calls that start
(e.g., \code{writeStartObject()}) and end
(e.g., \code{writeEndObject()}) a JSON element must be matched and
properly nested. Failure to follow this protocol results in the
generation of invalid JSON files.

\paragraph{Clients.}
To evaluate our approach on realistic usage scenarios of these
libraries, we collected ten open-source Java programs that use these
APIs. The clients used in our evaluation are widely-used programs such
as Hadoop/MapReduce (a distributed computing framework), ExoPlayer (an
Android media player), ConnectBot (secure shell client), Netflix
Hystrix (a fault tolerance library for distributed environments),
etc. These applications contain an average of 571 classes and 36,390
lines of Java code (equivalently, 56,114 Soot bytecode
instructions). Recall that we first slice the input program before we
run any of the verifiers (Section~\ref{sec:impl}). The effectiveness
of slicing varies across different benchmarks with the resulting
slices containing between 3-126 classes.

\subsection{ Results for CFPChecker}\label{subsec:results}

Table~\ref{table:res-cfpchecker} summarizes our main experimental results
for \toolname. As we can see from the "Output" column, two of the benchmarks 
(namely, ExoPlayer and ConnectBot) actually misuse at least one API. 
For these benchmarks (indicated with the $\lightning$ symbol), we also 
construct a correct variant (indicated without the $\lightning$ symbol) by manually repairing the original bug. 
We now summarize the key take-away lessons from this evaluation.  \\

\noindent
\emph{Verification results for correct benchmarks.}
\toolname\ is able to successfully verify  all benchmarks that 
correctly use the relevant API. On average, \toolname\ takes 9.3 minutes to 
verify each application, and its median verification time is 60.4 seconds. 
Most of the benchmarks require a significant number of refinement steps, with 
22.5 being the median number of iterations. \\

\noindent
\emph{Counterexamples for buggy benchmarks.} 
  As shown in Table~\ref{table:res-cfpchecker},
 \toolname\ reports three API protocol violations. Two of these violations are in ExoPlayer, which misuses both 
 the WifiLock and WakeLock libraries, and the other violation is in ConnectBot, which misuses WakeLock. 
Using the counterexamples reported by \toolname, we were able to identify the root causes of these errors. 
Interestingly, all three violations share the same root cause. In particular, 
ExoPlayer and ConnectBot both call the  \code{acquire}  method in \code{onStart}  and the corresponding \code{release} method in \code{onStop} of an Android Activity~\cite{android-activity}; however, they fail to 
release the lock in the \code{onPause} method. Since the Android framework may kill a paused activity  when there is memory pressure (see Figure~\ref{fig:android}}), the calls to \code{acquire} and \code{release} are not guaranteed to be matched. Thus, this  bug can result in resource leaks in the form of unintended battery usage.  One simple way to fix this issue is to move the \code{acquire} and \code{release} calls  to the \code{onResume} and \code{onPause} methods instead. In fact, a later version of the ConnectBot application fixes the bug in exactly this way; however,  \toolname\ identified a previously unknown issue in ExoPlayer.  \\

\begin{wrapfigure}{r}{0.4\textwidth}
\vspace{-12pt}
    \centering
      \includegraphics[width=0.35\textwidth]{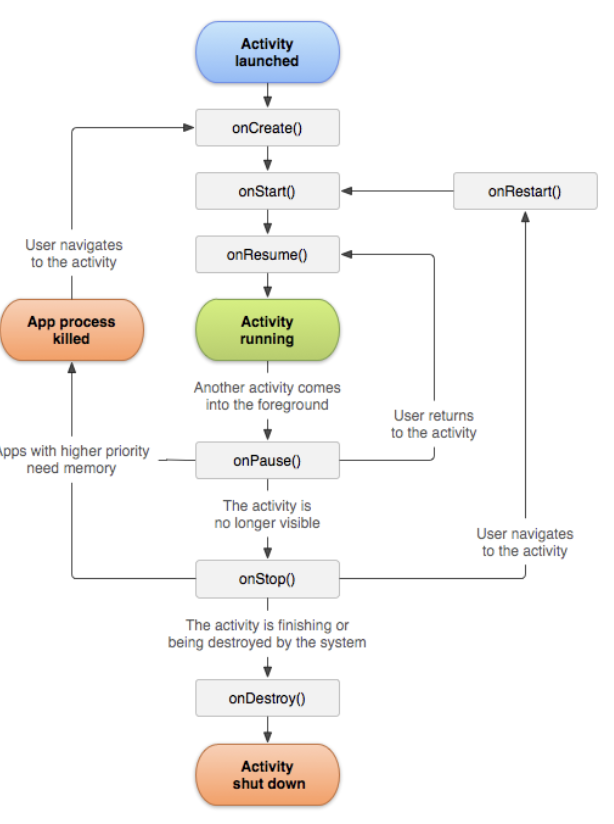}
    \caption{Android lifecycle callbacks}\label{fig:android}
\end{wrapfigure}
\FloatBarrier

\noindent
\emph{Summary.} As these experiments indicate, verifying the correct usage of context-free API protocols is of practical relevance 
in real-world applications. Our results demonstrate that \toolname\ is practical enough to verify the correct usage of context-free API protocols in widely-used Java applications and that it can provide useful counterexamples when the property is violated.

\subsection{Comparison with Baselines}\label{subsec:instr}
Since there is no existing tool for verifying correct usage of context-free  protocols, we cannot \emph{directly} compare our approach against existing baselines. Thus, we construct our own baselines using the following strategy: First, we instrument each program with suitable assertions that enforce correct API usage (as explained below). Then, we try to discharge these assertions using existing safety verifiers. In this section, we report on our experience implementing and evaluating these baselines using JayHorn~\cite{jayhorn}  and JavaPathFinder~\cite{jpf-svcomp} as  the assertion checking back-ends.  Note that JayHorn is a  state-of-the-art Java verification tool based on constrained Horn clause solvers, and JavaPathFinder is a mature model checking tool for Java developed by NASA.

\paragraph{Assertion instrumentation.}  As mentioned earlier, the goal of our instrumentation is to generate a program $P'$ such that $P'$  is free of assertion failures if and only if the original program $P$ obeys a given context-free API protocol. 
One obvious way to perform this instrumentation is to represent the API protocol using a push-down automaton (PDA) and then  introduce variables that keep track of the PDA's state and stack contents. In fact, this strategy has been used in prior work for performing {run-time checking} of correct API usage~\cite{mop-framework,jin2012javamop,meredith2010efficient}. However, since static techniques are typically not very good at reasoning about dynamically allocated data structures (e.g., arrays),  we instead manually perform   \emph{API-specific instrumentation} that  avoids introducing arrays whenever possible. For example,
for \code{ReentrantLock}, we only introduce an integer counter \code{c} that is incremented (resp. decremented) on calls to \code{lock} (resp. \code{unlock}). Then, to enforce the  protocol,  we assert that \code{c} is positive when \code{unlock} is called and that it is zero at the end. Using similar strategies, we can perform instrumentation using \emph{only integer variables} for all APIs except  one   (\code{JsonGenerator}).

Please note that the instrumentation strategy described above requires  human ingenuity and cannot be used to \emph{automatically} check arbitrary API protocols.  However, it is designed to be \emph{as favorable as possible} to assertion checking tools and represents the best possible scenario for existing verifiers.

\paragraph{Slicing and pre-processing.} Recall from Section~\ref{sec:impl} that \toolname\ incorporates a slicing step to enable better scalability. To ensure a fair comparison, we use the exact same slicing procedure before feeding the instrumented programs to the assertion checking tools. However,  even the slices contain several features that cannot be handled by at least one of these two tools. For instance, if we provide the generated slices to JayHorn as-is, it crashes on most benchmarks. Similarly, JavaPathFinder throws an exception whenever it encounters a call to a method whose source code is not available. Therefore, in order to use JayHorn and JavaPathFinder as our assertion-checking back-ends, we further manually simplified our benchmarks from Section~\ref{subsec:specs-bench} in a way that preserves the relevant  API usage-related behavior.  

\paragraph{Configurations of JavaPathFinder.} The JPF tool can be configured in several different ways. 
In this experiment, we use two configurations of JPF for the assertion-checking back-end. The first variant, henceforth called JPF-BugFinder, is a version of Java
Pathfinder that is configured with the default settings for
SVCOMP~\cite{jpf-svcomp}. Note that these settings are suitable for bug finding but not for verification.
To use JavaPathFinder as a verifier, we also consider a second variant where we do not restrict its search space. We refer to this variant as JPF-Verifier. Since JPF-BugFinder is \emph{not} a verifier, we only evaluate it on the buggy benchmarks.

\paragraph{Overall results.} The results of our comparison against these three baselines (JayHorn, JPF-BugFinder, and JPF-Verifier) are presented in Table~\ref{table:res-othertools}. The key take-away from this experiment is that none of the three baselines  are effective at successfully verifying (or finding bugs in) our experimental benchmarks despite manual simplification and instrumentation. In what follows, we describe the results for each of the three baselines in more detail.

\paragraph{Results for JayHorn.}
JayHorn verifies only 2 of the 15 benchmarks. For 7 benchmarks,
JayHorn reports a possible assertion violation, but is unable to
provide a counterexample. For the remaining 6 benchmarks, JayHorn
either fails to terminate within the 8-hour time limit or runs out of
memory. Surprisingly,  one of the benchmarks (Hystrix-1) that can
be verified by JayHorn uses the complex \code{JsonGenerator} API. We
conjecture that JayHorn can verify this benchmark more easily because
it does not involve recursion and all relevant API usage is confined
within a single method.

\paragraph{Results for JPF} When using JPF as a bug finder with the default SV-COMP settings, it fails to find the assertion violations in the three buggy benchmarks and reports them as safe. This result suggests that the API protocol violations  in the buggy benchmarks are non-trivial to find.  On the other hand, JPF-Verifier 
fails to terminate within the eight hour time limit on any benchmark, and it also fails to find the errors in the three buggy benchmarks.

\section{Related Work}
We now survey prior work related to  this paper and highlight their differences from our approach.

\paragraph{Typestate analysis.} Most prior work on checking correct API usage focuses on protocols that can be expressed as a regular language~\cite{strom1986typestate,typestate-fink, ball2006thorough}. This problem is  commonly known as \emph{typestate analysis}~\cite{strom1986typestate}, and researchers have proposed many different approaches to solve this problem ranging from language-based solutions~\cite{aldrich2009,bierhoff2007,deline2004typestates,garcia2014} to program analysis~\cite{typestate-fink,bodden2010efficient,bodden2012clara} and model checking~\cite{ball2001automatically,ball2006thorough} to  bug finding~\cite{joshi2008,yu2018symbolic} and run-time verification~\cite{trace-matching,mop-framework}. {Some prior works have also proposed various generalizations of typestate properties, such as multi-object protocols~\cite{multi-object1,beckman2011empirical}.

\paragraph{Run-time checking  for context-free properties.} 
There have been some proposals, particularly in the context of run
time techniques, for checking correct usage of APIs with
context-free specifications. In particular, these
techniques~\cite{pql,hawk,meredith2010efficient,jin2012javamop}
instrument the program with monitors that keep track of PDA states
and dynamically check for property violations. As shown in our experiments,
such an instrumentation-based approach does not work well for static verification.

\paragraph{Interface grammars.} Prior work has proposed 
\emph{interface grammars} for specifying the sequences of method invocations
that are allowed by a  library~\cite{hughes2008interface}. Given an interface grammar for a component, this technique
generates a stub that can be used to analyze clients of that component. While this work
addresses a somewhat different problem, their technique bears similarities to our
instrumentation-based baseline, which, as shown in our evaluation, does not work well in our setting.

\paragraph{CEGAR}
Similar to all CEGAR approaches~\cite{cegar1,cegar2, blast,lazy-abs,abs-from-proof,heizmann2018ultimate,sea-horn,hsf}, our method starts with a coarse  abstraction and iteratively 
refines it based on spurious counterexamples. However, our method differs from most CEGAR-based techniques
in that we abstract the program using a context-free grammar and perform refinement by adding new non-terminals and productions to the grammar. 

\paragraph{Abstracting programs with CFGs.} Similar to our approach, prior work on has explored abstracting programs using context-free grammars.  For example, Long et al~\cite{long-rupak}  use CFG inclusion checking to  prove assertions in concurrent programs; however, their approach does \emph{not} refine the program's CFG abstraction. Instead, they use a CEGAR approach to solve the  CFG inclusion checking problem through  a sequence of increasingly more precise regular approximations. Furthermore, since they address a different problem, their CFG abstraction is quite different from ours. Another related approach in this space is the work by~\citet{bounded-under} which  also abstracts recursive multi-threaded programs
with a context-free grammar. In contrast to our work, they
\emph{under-approximate} the reachable state space of recursive
multi-threaded programs by generating a succession of bounded
languages that under-approximate the program's CFG.

\paragraph{Interpolants.} Similar to many CEGAR-based
techniques~\cite{abs-from-proof,impact,sea-horn,mcmillan2005applications},
our method also uses \emph{Craig interpolation} to learn new
predicates when a spurious counterexample is discovered.  Given an
unsatisfiable formula $\phi \land \psi$, a Craig interpolant is
another formula $\chi$ such that $\phi \Rightarrow \chi$ is valid and
$\psi \land \chi$ is unsatisfiable. Prior work has proposed many
variants of Craig interpolation, including sequence
interpolants~\cite{abs-from-proof}, tree
interpolants~\cite{tree-interp}, nested
interpolants~\cite{nested-interp}, and DAG
interpolants~\cite{dag-interp}. In this paper, we leverage the notion
of nested interpolants introduced in~\citet{nested-interp} to infer
useful predicates for recursive procedures; however, our refinement
procedure uses these nested interpolants in a very different way.

\paragraph{Control flow refinement.} 
Our refinement technique bears similarities to prior work on
control-flow
refinement~\cite{control-flow-ref1,control-flow-ref2,control-flow-ref3,control-flow-ref4}. Similar to \toolname, these techniques clone program locations 
in order to exclude infeasible paths from
their program abstraction.  However, all of these techniques abstract the 
program using a regular language, and, with the exception of~\citet{control-flow-ref4},
they apply control-flow refinement within a
single procedure and only inside loops. On the other hand, 
\citet{control-flow-ref4} refines cost equations rather than the
 program abstraction. In contrast to all of these techniques, our technique refines the CFG abstraction, performs cloning
inter-procedurally, and supports arbitrary recursion.

\paragraph{Directed proof generation.}
Directed proof
generation (DPG) techniques simultaneously maintain an
under- and an over-approximation of the program and evolve them in a
synergistic way~\cite{mcveto}. Specifically, the under-approximation
is used to find feasible counterexamples and learn new predicates
which refine the over-approximation. Conversely, the
over-approximation is used to generate proofs and guides
counterexample search to paths that are more likely to fail. Similar
to our technique, DPG-like approaches~\cite{synergy,dash,smash,mcveto}
annotate their control-flow representation with logical predicates and
clone program locations. Our approach differs from these techniques in the way it
discovers potential counterexamples and  new
predicates. In particular, \toolname\ performs an inclusion check between
two context-free languages in order to discover a potential API
violation and uses interpolation to discover new predicates. In contrast,
DPG techniques use a combination of graph reachability and test-case
generation.

\paragraph{Equivalence of context-free languages.} Our approach leverages prior work on checking containment between two context-free languages~\cite{grammar-comp,korenjak1966simple,olshansky1977direct,harrison1979equivalence}. While checking inclusion between arbitrary context-free languages is known to be undecidable, prior work has studied decidable fragments, such as LL$(k)$ grammars~\cite{olshansky1977direct}. Our implementation makes use of the algorithm by~\citet{grammar-comp}, which in turn extends prior algorithms for LL grammars. While our technique is orthogonal to checking context-free language containment, it would directly benefit from advances and new algorithms that address this problem.

\paragraph{CFL reachability.}
CFL reachability techniques represent inter-procedural control flow
using a graph representation and then filter out paths that do not
conform to valid call-return structures~\cite{reps1995precise}.  This
formulation has been used to express several fundamental program
analyses, such as context-sensitive pointer
analysis~\cite{multi-cfl-3, multi-cfl-4}. However, adding another
level of sensitivity (e.g., field-sensitivity,) requires solving two
separate CFL reachability problems on the same execution path, which
is known to be undecidable~\cite{reps-undecidable}; hence many
techniques over-approximate one of the two CFL reachability
problems~\cite{multi-cfl-1, multi-cfl-2, multi-cfl-3, multi-cfl-4,
  multi-cfl-5, multi-cfl-6, dyck-graph-simpl} or propose a more
precise generalization of CFL
reachability~\cite{hao-tang,lcl-reach}. Similar to these techniques,
we also need to reason about two context-free properties, namely
matching call-return statements and matching between calls to API
methods. However, this work addresses a somewhat different problem:
instead of filtering out execution paths that do not belong to both
context-free languages, our technique verifies that \emph{every} API
sequence generated by an execution path with a valid call-return
structure belongs to the context-free specification.

\paragraph{Visibly pushdown automata}
Many model checking techniques use variants of pushdown automata, such
as \emph{visibly pushdown automata} (VPAs) or \emph{nested word
  automata} (which are equally expressive), to  reason about
inter-procedural control flow~\cite{chen2002mops, esparza2003model,
  lazy-abs,alur2004visibly}.  Visibly pushdown and nested word
automata are less expressive compared to PDAs;
however, they enjoy various decidability and closure properties for
operations like intersection and complement.
However, VPAs cannot capture two separate context-free properties
on the same execution path, which  is
required by our technique.

There have been some theoretical studies that extend VPAs to use
multiple stacks~\cite{2vpa,multistack-vpa1,multistack-vpa2}, and such
multi-stack VPAs are significantly more expressive compared to
standard VPAs.  For example, 2-VPAs~\cite{2vpa} (i.e., VPAs with two
stacks) accept some context-sensitive languages that are not
context-free and some context-free languages that are not accepted by
any VPA. We believe that it would be possible to solve the problem
addressed in this paper using 2-VPAs, however, the emptiness problem
for 2-VPAs is also undecidable.

\section{Conclusion}

We presented a technique for verifying the correct usage of context-free API protocols. Our approach abstracts the program as a context-free grammar  representing feasible API call sequences and checks whether this CFG is contained inside the specification CFG. 
Our method follows the CEGAR paradigm and  lazily refines the CFG by introducing new productions and non-terminals that represent clones of methods and program locations.

We  implemented the proposed method in a tool called \toolname\ and performed an experimental evaluation on 10 widely-used Java applications that utilize at least one of 5 popular APIs with context-free specifications.  Our evaluation shows that \toolname\ can verify all correct usage patterns while finding counterexamples for the buggy clients. We also implement and evaluate three baselines that reduce this problem to assertion checking and then use oft-the-shelf safety checking tools to discharge these assertions. Our experience with these baselines suggests that our method is more amenable to automation than alternative approaches that reduce the problem to assertion checking.

\section*{Acknowledgments}

We would like to thank our shepherd Pierre Ganty, the anonymous
reviewers, Kenneth McMillan, Swarat Chaudhuri, and the members of the
UToPiA group for their insightful feedback. This work is supported in
part by NSF Award \#1453386 and DARPA Award \#FA8750-20-C-0208.

\bibliography{main}

\iffull{
  \appendix
  \section{Appendix: Correctness of Program Instrumentation}

In this section we prove the correctness of the program
instrumentation presented in Section~\ref{sec:instr}. The proof of
Theorem~\ref{thm:equi-safe} follows from lemma~\ref{lemma:equi-1}.

\begin{lemma}\label{lemma:equi-1}
  If we have \ $\typeor, \cfgspec \vdash \program \hookrightarrow
  \program'$, then for every trace $\ptrace \in \program$ and
  $\instspec \in \inst(\cfgspec, \ptrace)$ there exists a trace
  $\ptrace' \in \program'$ such that $\project(\ptrace, \instspec) =
  \project(\ptrace', \instspec)$.
\end{lemma}
\begin{proof}
  We assume the existence of a map $\Lambda: Loc \rightarrow Loc$,
  which given a program location in $\program$ it returns the
  corresponding location in $\program'$. For \code{api_call}
  statements, $\Lambda$ returns the location of the instrumented
  \code{if}-then-\code{else} statement as presented in
  Figure~\ref{fig:instr}.

  Now, given a trace in $\program$ of the form $\ptrace = \langle s_1,
  \sigma_1 \rangle ... \langle s_n, \sigma_n \rangle$ and a
  specification instantiation $\instspec \in \inst(\cfgspec,
  \ptrace)$, we show how to construct a trace $\ptrace' \in \program'$
  with the same trace projection as $\ptrace$. For the rest of the
  proof, we will only describe the statements we append in $\ptrace'$
  without showing the corresponding program states since each program
  state can be obtained by using the operational semantics (i.e.,
  $\Downarrow$ relation) of the preceding (statement, prog. state)
  pair in $\ptrace'$. In what follows, we assume that: 1. $\vec{v}$
  are the concrete values that instantiate all the wildcards in the
  specification, i.e., $\instspec = \cfgspec[\vec{v}/\vec{w}]$.
  2. For any $\langle s_i, \sigma_i \rangle \in \ptrace$ we assume
  that $\sigma_{i+1}$ is the result of executing $s_i$ under
  $\sigma_i$, i.e., $\langle s_i, \sigma_i \rangle \Downarrow
  \sigma_{i+1}$ and 3. For any statement $s_i$ in $\ptrace$ we assume
  that $l_i$ represents its location in $\program$.

  First, for every wildcard field $w_i$ in $\program'$ we append the
  assignment \code{$w_i$ = $v_{i}$} to $\ptrace'$, where $v_{i}$ is
  the i-th element of vector $\vec{v}$. Next, we iterate over every
  pair $\langle s_i, \sigma_i \rangle$ in $\ptrace$ and update
  $\ptrace'$ according to the following rules:
  \begin{enumerate}
     \item If $s_i$ is of the form \code{api\_call m($\vec{v1}$)},
       then for every \code{if (guard($t_i, s_i$)) $t$} in the
       statement at $\Lambda[l_i]$ update $\ptrace'$ as follows: If we
       have that \code{guard($t_i, s_i$)} evaluates to true for any
       pair of $s_i, t_i$ under $\sigma_i$\footnote{This can be easily
         determined by checking whether each value in $\vec{v1}$
         equals to value that instantiates the corresponding wildcard
         in $\vec{v}$.}, then append $t_i$ to $\ptrace'$. Otherwise,
       append \code{skip} to $\ptrace'$.
     \item Otherwise, append $s_i$ to $\ptrace'$.
  \end{enumerate}

  Intuitively, the first step, i.e., the field initialization, ensures
  that the wildcard fields in $\program'$ have the same value as the
  instantiated specification $\instspec$, whereas, the second step
  ensures that the branches taken in $\ptrace'$ are semantically
  consistent with the initialization of the wildcard fields. It is
  easy to see that $\ptrace'$ can be generated by $\program'$ and that
  $\project(\ptrace, \instspec) = \project(\ptrace', \instspec)$ since
  the only guards that will evaluate to true are the ones where
  \emph{all} the program variables of the API call equal to some value
  in $\vec{v}$.
\end{proof}

  \section{Appendix: Soundness and Progress Proofs}

In this section, we prove Theorems~\ref{thm:soundness}
and~\ref{thm:completeness}, which state the soundness and progress
of our approach respectively. Theorem~\ref{thm:soundness} follows from
Lemmas~\ref{lem:deriv2seq}
and~\ref{lem:sound}. Theorem~\ref{thm:completeness} follows from
Lemmas~\ref{lem:deriv2seq} and~\ref{lem:compl}.

\begin{definition}
  {\bf{PCFA State Sequence:}} A PCFA state sequence is a tuple of the
  form $(s_0 ... s_n, \rightsquigarrow)$, where each $s_i$ is a PCFA
  state of some method $m \in \program$ and $\rightsquigarrow$ is a
  nesting relation over the indices~$[0, n]$.
\end{definition}

\begin{definition}\label{def:feas-state-seq}
  \begin{sloppypar}
  {\bf{$\program$-feasible State Sequence:}} We say that a state
  sequence is $\program$-feasible and write $\program \vdash (s_0
  ... s_n, \rightsquigarrow)$ if and only if the following hold:
  \begin{enumerate}
    \item For all $i,\ j$ such that $i \rightsquigarrow j$, there
      exists an edge $(s_i, \code{call m}, s_j)$ in
      $\program$. Furthermore, we have $s_{i+1} \in
      \initsof(\program[m])$, $s_{j-1} \in \finsof(\program[m])$, and
      $SAT(\predof(s_i) \land \predof(s_{j-1}) \land \predof(s_j))$.
    \item For all $i \in [0, n-1]$ for which $\rightsquigarrow$ is
      undefined, there exists an edge $(s_i, \sigma, s_{i+1})$ in
      $\program$.
  \end{enumerate}
  \end{sloppypar}
\end{definition}

\begin{definition}\label{def:feas-path}
  {\bf{$\program$-feasible Execution Path:}} An execution path
  $(\sigma_0 ... \sigma_n, \rightsquigarrow)$ is $\ \program$-feasible
  through state sequence $(s_0 ... s_{n+1}, \rightsquigarrow)$,
  denoted as $\program, (s_0 ... s_{n+1}) \vdash (\sigma_0
  ... \sigma_n, \rightsquigarrow)$, if and only if $\program \vdash
  (s_0 ... s_{n+1}, \rightsquigarrow)$ and for every statement
  $\sigma_i$:
  \begin{enumerate}
    \item If $i \rightsquigarrow j$, then $\sigma_i$ is a call
      statement, $\sigma_{j-1}$ is the matching return statement, and
      $(s_i, \sigma_i, s_j)$ is an edge in the PCFA.
    \item If $\sigma_i$ is not a call or a return statement, then
      $(s_i, \sigma_i, s_{i+1})$ is an edge in the PCFA.
  \end{enumerate}
\end{definition}

\paragraph{Oracle $\mathit{deriv2seq}$}
In a similar manner as $\mathit{derivation2path}$, we assume the
existence of a similar oracle named $\mathit{deriv2seq}$ which given a
derivation $d \in \cfgprog$ it returns its corresponding PCFA sequence
$(s_0 ... s_n, \rightsquigarrow)$. This can be easily derived from a
derivation since every non-terminal in $\cfgprog$ is associated with a
singe PCFA state.

\begin{lemma}\label{lem:deriv2seq}
  A derivation $d$ belongs to $\cfgprog$ if and only if $\ \program
  \vdash \mathit{deriv2seq}(d)$.
\end{lemma}
\begin{proof}
  This follows immediately from the construction of $\cfgprog$. Let's
  assume that for a derivation $d \in \cfgprog$ we have that
  $\mathit{deriv2seq}(d) = (s_0 ... s_n, \rightsquigarrow)$.

  ($\Longrightarrow$) For this direction, we will prove that $\program
  \vdash (s_0 ... s_n, \rightsquigarrow)$. In order for
  $\mathit{deriv2seq}$ to return this sequence the following must
  hold:
  \begin{itemize}
    \item For all $i, j$ such that $i \rightsquigarrow j$, there must
      exist a rule in $\cfgprog$ of the form $\mathcal{S}_{i\phi}
      \rightarrow \mathcal{M}_{\phi'} \mathcal{S}_{j\phi}$, where
      $\phi, \phi'$ refer to the method clones of the caller and the
      callee. Additionally we have that $s_{i+1} \in
      \initsof(\program[m])$, $s_{j-1} \in \finsof(\program[m])$, and
      $SAT(\predof(s_i) \land \predof(s_{j}) \land
      \predof(s_{j-1})))$. By construction of $\cfgprog$, this implies
      that there exists an edge of the form $(s_i, \code{call m},
      s_j)$ in $\program$.
    \item For all $i$ that $\rightsquigarrow$ is undefined, there must
      exist a rule in $\cfgprog$ of the form $\mathcal{S}_{i\phi}
      \rightarrow \mathcal{S}_{i+1\phi}$ or $\mathcal{S}_{i\phi}
      \rightarrow t\ \mathcal{S}_{i+1\phi}$ where $t$ is a
      terminal. This also implies that there exist an edge of the form
      $(s_i, \sigma, s_{i+1})$ in $\program$.
  \end{itemize}
  Now, it is easy to see that both these conditions satisfy
  Definition~\ref{def:feas-state-seq}, therefore we have that
  $\program \vdash (s_0 ... s_n, \rightsquigarrow)$.

  ($\Longleftarrow$) We now prove that if $\program \vdash (s_0
  ... s_n, \rightsquigarrow)$, then there exists $d \in \cfgprog$ such
  that $\mathit{deriv2seq(d)} = (s_0 ... s_n, \rightsquigarrow)$.

  This follows by the construction of $\cfgprog$ and
  Definition~\ref{def:feas-state-seq}. Specifically, the first
  condition of~\ref{def:feas-state-seq} implies that there exists at
  least one production in $\cfgprog$ of the form $\mathcal{S}_{i\phi}
  \rightarrow \mathcal{M}_{\phi'}\ \mathcal{S}_{j\phi}$ for some
  clones $\phi$ and $\phi'$. Whereas, the second condition implies
  that there exists a rule of the form $\mathcal{S}_{i\phi}
  \rightarrow \mathcal{S}_{i+1\phi}$ or $\mathcal{S}_{i\phi}
  \rightarrow t\ \mathcal{S}_{i+1\phi}$ for some method clone
  $\phi$. Therefore, we have that there exists at least one derivation
  $d \in \cfgprog$ for which $\mathit{deriv2seq(d)} = (s_0 ... s_n,
  \rightsquigarrow)$.
\end{proof}

\begin{lemma}\label{lem:compl}
  Let $\ \program'$ be the program returned by procedure
  $\textsc{Refine}$ for spurious counterexample $(\sigma_0
  ... \sigma_n, \rightsquigarrow)$ and sequence of nested interpolants
  $\mathcal{I} = [I_0, ..., I_{n+1}]$. Then we have that there does
  not exist state sequence $(s_0 ... s_{n+1}, \rightsquigarrow)$ such
  that $\ \program', (s_0 ... s_{n+1}) \nvdash (\sigma_0 ... \sigma_n,
  \rightsquigarrow)$.
\end{lemma}
\begin{proof}
  Let's assume that $\program$ is the program before the
  refinement. To simplify the proof we make the following assumptions:
  1. $(\sigma_0 ... \sigma_n, \rightsquigarrow)$ is the first
  counterexample, i.e., each PCFA state has true as a predicate 2. a
  program location appears only once in the counterexample and 3. the
  last statement of \code{main} is skip. The third assumption just
  converts the input program to a normal form. Later we show how to
  generalize the proof so it does not require the first two
  assumptions.

  \begin{sloppypar}
  By Definition~\ref{def:feas-path} and assumptions 1 and 2, there
  exists a \emph{singe} state sequence $(s_0 ... s_{n+1},
  \rightsquigarrow)$ such that $\program, (s_0 ... s_{n+1}) \vdash
  (\sigma_0 ... \sigma_{n}, \rightsquigarrow)$ (i.e., the spurious
  counterexample is $\program$-feasible). Now, by the way procedure
  $\textsc{Refine}$ works, we have that states $s_0$ through $s_{n+1}$
  have been replaced by states $s_0',s_0'',...,s_{n+1}',s_{n+1}''$ in
  $\program'$ such that: $\locof(s_i') = \locof(s_i'') = \locof(s)$
  and $\predof(s_i') = I_i$, $\predof(s_i'') = \neg I_i$. We now show
  that for any state sequence of the form $(\bar{s}_0
  ... \bar{s}_{n+1}, \rightsquigarrow)$ where $\bar{s}_i \in \{s',
  s''\}$ we have that: $\program', (\bar{s}_0 ... \bar{s}_{n+1})
  \nvdash (\sigma_0 ... \sigma_{n+1}, \rightsquigarrow)$.
  \end{sloppypar}

  Now recall that procedure $\textsc{Refine}$ removes all the edges
  $(s_1, \sigma, s_2)$ for which $UNSAT(sp(\sigma, \predof(s_1)) \land
  \predof(s_2))$ holds. Since $I_{n+1} = false$ we have that the
  counterexample cannot be feasible through any state sequence that
  ends with $s_{n+1}'$ and from the definition of nested interpolants
  we also get that $UNSAT(sp(\sigma_n, I_n))$. Hence, the
  counterexample cannot be feasible through any state sequence of the
  form $(\bar{s}_0' ... s_n' s_{n+1}'', \rightsquigarrow)$. Last, we
  prove that $\program', (\bar{s}_0 ... s_{n}'' s_{n+1}'') \nvdash
  (\sigma_0 ... \sigma_n, \rightsquigarrow)$ for any combination of
  $\bar{s}_i$ which concludes the proof.

  \begin{sloppypar}
  The proof is by induction on the length of the counterexample:
  \begin{itemize}
    \item {\bf {Base case}}: Here we need to prove that $\program', (s_0''
      s_1'') \nvdash (\sigma_0, \rightsquigarrow)$. Since $I_0$ is
      true, we have $UNSAT(sp(\sigma_0, Pred(s_0'')))$ therefore the
      execution path is not feasible through state sequence $(s_0''
      s_1'')$
    \item {\bf {Inductive Step}}: By inductive hypothesis we have that
      $\program', (\bar{s}_0 ... s_{n+1}'') \nvdash (\sigma_0
      ... \sigma_n, \rightsquigarrow)$ for any $\bar{s}_i$. We will
      now prove that $\program', (s_0' ... s_{n+1}' s_{n+2}'')
      \nvdash (\sigma_0 ... \sigma_{n+1})$. Here we have two cases:
      \begin{enumerate}
        \item $\sigma_{n+1}$ is not a return statement. By the
          definition of nested interpolants we have that
          $sp(\sigma_{n+1}, \predof(s_{n+1}')) \Rightarrow
          \predof(s_{n+2}')$, this implies $UNSAT(sp(\sigma_{n+1},
          \predof(s_{n+1}')) \land \predof(s_{n+2}''))$. Therefore,
          $\textsc{Refine}$ removes the edge $(s_{n+1}, \sigma_{n+1},
          s_{n+2}'')$ which make the execution path infeasible through
          this predicate sequence.
        \item $\sigma_{n+1}$ is a return statement. By the definition
          of nested interpolants we have that $sp(\sigma_{n+1},
          \predof(s_{n+1}') \land \predof(s_{j}')) \Rightarrow
          \predof(s_{n+2}')$ for $j \rightsquigarrow n+2$. This
          implies that $UNSAT(sp(\sigma_{n+1}, \predof(s_{n+1}') \land
          \predof(s_{j}')) \land \predof(s_{n+2}''))$, hence by
          Definition~\ref{def:feas-path} we have that the execution
          path is infeasible through this state sequence.
      \end{enumerate}
  \end{itemize}
  \end{sloppypar}

  Now, if we lift assumptions 1 and 2 from earlier this means that
  each state $s_i$ from the feasible state sequence will be replaced
  by multiple new states in $\program'$. Recall though that procedure
  $\textsc{Refine}$ will clone all the states in $\program$ with the
  same location as $s_i$ and the resulting clones in $\program'$ will
  have one of the complete cubes as their predicate. This means that
  for a given interpolant $I_i$ a clone $s_i'$ in $\program'$ will
  contain either $I_i$ or $\neg I_i$, which implies that
  $\predof(s_i') \Rightarrow I_i$ or $\predof(s_i') \Rightarrow \neg
  I_i$. Therefore, the proof in the general case is similar to the one
  above except that one would have to consider \emph{all} the clones
  of $\program'$ that contain the predicate $\neg I_i$.
\end{proof}

\begin{lemma}\label{lem:sound}
  \begin{sloppypar}
  If $\ \program'$ is the program returned by procedure
  $\textsc{Refine}$ for sequence of nested interpolants $\mathcal{I} =
  [I_0, ..., I_{n+1}]$, then there exists state sequence $(s_0
  ... s_m, \rightsquigarrow)$ such that $\program', (s_0 ... s_m)
  \vdash (\pi, \rightsquigarrow)$ for every feasible execution path in
  $\program$.
  \end{sloppypar}
\end{lemma}
\begin{proof}
  Here we make the same assumptions as in Lemma~\ref{lem:compl}, but
  as before the proof generalizes for the same reasons.

  Let's assume that $(\sigma_0 ... \sigma_n, \rightsquigarrow)$ is the
  spurious counterexample that triggered the refinement. Again, by
  Definition~\ref{def:feas-path} and assumptions 1 and 2 from
  Lemma~\ref{lem:compl}, there exists a \emph{singe} state sequence
  $(s_0 ... s_{n+1}, \rightsquigarrow)$ such that $\program, (s_0
  ... s_{n+1}) \vdash (\sigma_0 ... \sigma_{n}, \rightsquigarrow)$
  (i.e., the spurious counterexample is $\program$-feasible). Now, let
  $s_0', s_0'', ... s_n',s_n''$ be the states that replace each state
  $s_i$ in $\program$ where $\predof(s_i') = I_i$ and $\predof(s_i'')
  = \neg I_i$. Now, in order for an execution path to be eliminated by
  procedure $\textsc{Refine}$ it must have a statement that labels an
  edge whose source or destination state is one of the states $s_i$
  (the rest of the graph does not change). We prove this is impossible
  by contradiction:

  Let's assume that there exists an execution path $(\pi = \sigma_0
  ... \sigma_m, \rightsquigarrow)$ for which there does not exist
  state sequence $(s0 ... s_{m+1}, \rightsquigarrow)$ in $\program'$
  such that $\program', (s0 ... s_{m+1}, \rightsquigarrow) \nvdash
  (\pi, \rightsquigarrow)$ and also one of the statements $\sigma_j$
  in $\pi$ labels an edge whose target state\footnote{The case where
    it labels an edge whose source state is one from the
    counterexample is similar.} is one of the cloned states $s_{j+1}'$
  or $s_{j+1}''$. Now in order for this execution path to not be
  feasible we must have one of the following:
  \begin{itemize}
    \item If $\sigma_j$ is not a return statement, then we must have
      that $sp(\sigma_j, \predof(s_{j})) \land \predof(s_{j+1}'))$ and
      $sp(\sigma_j, \predof(s_{j})) \land \predof(s_{j+1}''))$ which is
      impossible since $\predof(s_{j+1}') = I_{j+1}$ and
      $\predof(s_{j+1}'') = \neg I_{j+1}$.
    \item If $\sigma_j$ is a return statement, then we must have that
      $sp(\sigma_j, \predof(s_{j}) \land \predof(s_{i})) \land
      \predof(s_{j+1}'))$ and $sp(\sigma_j, \predof(s_{j}) \land
      \predof(s_{i})) \land \predof(s_{j+1}''))$ for $i
      \rightsquigarrow j+1$. Which again it is impossible for the same
      reason as above.
  \end{itemize}
\end{proof}

}\fi

\end{document}